\newif\ifAnon\Anontrue
\newif\ifNotes\Notestrue
\newif\ifCCS\CCStrue
\newif\ifRevision\Revisiontrue %
\newif\ifRevision\Revisionfalse
\newif\ifCamera\Cameratrue

\documentclass[sigconf]{acmart}

\usepackage{cuted}
\usepackage{bm}
\usepackage{tikz}
\usepackage{amsmath}
\usepackage{multirow}

\usepackage{graphicx}
\usepackage{xcolor}
\usepackage{fancyvrb}
\usepackage{caption}
\usepackage{booktabs}
\usepackage{mathtools}
\usepackage{xspace}
\usepackage{amsthm}
\usepackage{rotating}
\usepackage{filecontents}

\usepackage{algpseudocode}
\usepackage{algorithm}
\usepackage{array,booktabs,ragged2e}

\usepackage{subcaption}

\usepackage{hyperref}
\usepackage[capitalise,nameinlink,noabbrev]{cleveref}
\newcolumntype{R}[1]{>{\RaggedLeft\arraybackslash}p{#1}}

\makeatletter
\renewcommand\theHALG@line{\thealgorithm.\arabic{ALG@line}}
\makeatother

\usepackage{amssymb}%
\usepackage{pifont}%

\usepackage{enumitem}
\usepackage[normalem]{ulem}

\definecolor{citecolor}{rgb}{0.65,0,0}
\definecolor{linkcolor}{rgb}{0.,0.,0.}
\definecolor{urlcolor}{rgb}{0,0,0.65}
\hypersetup{pdfpagemode=UseNone,pdfstartview=FitH,colorlinks=true, linkcolor=linkcolor, urlcolor=urlcolor, citecolor=citecolor}

\newcommand{\Setup}{\mathsf{Setup}}
\newcommand{\Init}{\mathsf{Init}}
\newcommand{\Load}{\mathsf{Load}}
\newcommand{\Store}{\mathsf{Store}}

\newcommand{\scm}{\mathsf{scm}}
\newcommand{\nonce}{\mathsf{nonce}}

\newcommand{\UpdateState}{\mathsf{UpdateState}}
\newcommand{\InitState}{\mathsf{InitState}}
\newcommand{\ReadState}{\mathsf{GetState}}

\newcommand{\genSymKey}{\mathsf{genEncKey}}
\newcommand{\genAsymKey}{\mathsf{genSignKey}}

\newcommand{\mech}{\mathcal{M}_\query}
\newcommand{\cmark}{\ding{51}}%

\newcommand{\msg}{\mathsf{msg}}
\newcommand{\cd}{D}
\newcommand{\enccd}{\overline{\cd}}
\newcommand{\budget}{B}
\newcommand{\ReplyQueries}{\mathsf{ReplyQuery}}
\newcommand{\query}{q}
\newcommand{\answer}{a}
\newcommand{\id}{id}

\newcommand{\datakey}{k_\cd}
\newcommand{\statekey}{k_s}
\newcommand{\vk}{\mathsf{vk}}
\newcommand{\sk}{\mathsf{sk}}
\newcommand{\SCM}{\mathsf{SCM}}
\newcommand{\ST}{\mathsf{ST}}
\newcommand{\spec}{\mathsf{DPSpec}}

\newcommand{\dpsystem}{\textsc{ElephantDP}\xspace}

\newcommand{\naivedp}{\textsc{NaiveDP}\xspace}

\newcommand{\curstate}{\mathsf{curstate}}

\newcommand{\newstate}{\mathsf{newstate}}

\newcommand{\Encrypt}{\mathsf{Encrypt}}
\newcommand{\Decrypt}{\mathsf{Decrypt}}
\newcommand{\hash}{\mathsf{Hash}}
\newcommand{\Sign}{\mathsf{Sign}}
\newcommand{\Verify}{\mathsf{Verify}}
\newcommand{\statetext}{\text{`state'}\xspace}
\newcommand{\datatext}{\text{`data'}\xspace}

\newcommand{\encs}{\overline{s}}

\newcommand{\eg}{{e.g.,}\xspace}
\newcommand{\ie}{{i.e.,}\xspace}
\newcommand{\RMSE}{\ensuremath\mathsf{RMSE}\xspace}

\definecolor{fuchsia}{HTML}{FF00FF}
\ifNotes
  \newcommand{\colorednote}[2]{\leavevmode\unskip\space\textcolor{#1}{#2}\xspace}
\else
  \newcommand{\colorednote}[2]{\leavevmode\unskip\relax}
\fi

\ifRevision
\definecolor{ForestGreen}{HTML}{228B22}
\newcommand\remove[1]{{\textcolor{ForestGreen}{\sout{#1}}}}

\else
\newcommand\remove[1]{}

\fi

\newcommand{\parhead}[1]{\vspace{3pt plus 1pt minus 1pt}\par\noindent\textbf{#1}\hspace{.4em plus .2em minus .2em}}
\renewcommand{\paragraph}[1]{\parhead{#1}}

\AtBeginDocument{%
}

\copyrightyear{2024}
\acmYear{2024}
\setcopyright{rightsretained}
\acmConference[CCS '24]{Proceedings of the 2024 ACM SIGSAC Conference on Computer and Communications Security}{October 14--18, 2024}{Salt Lake City, UT, USA}
\acmBooktitle{Proceedings of the 2024 ACM SIGSAC Conference on Computer and Communications Security (CCS '24), October 14--18, 2024, Salt Lake City, UT, USA}
\acmDOI{10.1145/3658644.3670281}
\acmISBN{979-8-4007-0636-3/24/10}

\makeatletter
\gdef\@copyrightpermission{
 \begin{minipage}{0.3\columnwidth}
   \href{https://creativecommons.org/licenses/by-nd/4.0/}{\includegraphics[width=0.90\textwidth]{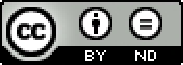}}
 \end{minipage}\hfill
 \begin{minipage}{0.7\columnwidth}
   \href{https://creativecommons.org/licenses/by-nd/4.0/}{This work is licensed under a Creative Commons Attribution-NoDerivs International 4.0 License.}
 \end{minipage}
 \vspace{5pt}
}
\makeatother

\begin{document}

\title{Elephants Do Not Forget: \\ Differential Privacy with State Continuity for Privacy Budget}

\author{Jiankai Jin}
\orcid{0009-0009-1008-482X}
\affiliation{%
  \institution{The University of Melbourne}
  \city{Melbourne}
  \country{Australia}}
\email{jiankaij@student.unimelb.edu.au}

\author{Chitchanok Chuengsatiansup}
\orcid{0000-0002-0329-2681}
\affiliation{%
  \institution{The University of Melbourne}
  \city{Melbourne}
  \country{Australia}}
\email{c.chuengsatiansup@unimelb.edu.au}

\author{Toby Murray}
\orcid{0000-0002-8271-0289}
\affiliation{%
  \institution{The University of Melbourne}
  \city{Melbourne}
  \country{Australia}}
\email{toby.murray@unimelb.edu.au}

\author{Benjamin~I.~P. Rubinstein}
\orcid{0000-0002-2947-6980}
\affiliation{%
  \institution{The University of Melbourne}
  \city{Melbourne}
  \country{Australia}}
\email{benjamin.rubinstein@unimelb.edu.au}

\author{Yuval Yarom}
\orcid{0000-0003-0401-4197}
\affiliation{%
  \institution{Ruhr University Bochum}
  \city{Bochum}
  \country{Germany}}
\email{Yuval.Yarom@rub.de}

\author{Olga Ohrimenko}
\orcid{0000-0002-9735-0538}
\affiliation{%
  \institution{The University of Melbourne}
  \city{Melbourne}
  \country{Australia}}
\email{oohrimenko@unimelb.edu.au}

\renewcommand{\shortauthors}{Jiankai Jin et al.}

\begin{abstract}
Current implementations of differentially-private (DP) systems 
either lack support to track the global privacy budget consumed on a dataset, 
or fail to faithfully maintain the state continuity of this budget.
We show that failure to maintain a privacy budget enables an adversary to mount replay, 
rollback and fork attacks --- obtaining answers to many more queries than what a secure system would allow.
As a result the attacker can reconstruct secret data that DP aims to protect ---
even if DP code runs in a Trusted Execution Environment (TEE).
We propose \dpsystem, a system that aims to provide the same guarantees 
as a trusted curator in the global DP model would, albeit set in an untrusted environment.
Our system relies on a state continuity module to provide protection for the privacy budget and
a TEE to faithfully execute DP code and update the budget.
To provide security, our protocol makes several design choices including the content of the persistent state
and the order between budget updates and query answers.
We prove that \dpsystem provides \emph{liveness} (\ie the protocol can restart from a correct state 
and respond to queries as long as the budget is not exceeded) and \emph{DP confidentiality} 
(\ie an attacker learns about a dataset as much as it would from interacting with a trusted curator). 
Our implementation and evaluation of the protocol use Intel SGX as a TEE to run the DP code 
and a network of TEEs to maintain state continuity.
Compared to an insecure baseline,
we observe 1.1--3.2$\times$ overheads and lower relative overheads 
for complex DP queries.
\end{abstract}

\ccsdesc[500]{Security and privacy~Privacy protections}
\ccsdesc[500]{Security and privacy~Security in hardware}

\keywords{Differential Privacy, Trusted Execution Environment, State Continuity}

\maketitle

\section{Introduction}
Differential privacy (DP) is becoming the de facto framework for maintaining
data privacy when releasing statistics or machine learning models.
Its minimal assumptions on the adversary's background knowledge
and quantification of accumulated privacy loss across multiple data releases 
have contributed to several practical deployments (\eg by Apple, Google, Meta, Microsoft, 
and the U.S.~Census Bureau~\cite{tang2017privacy,bittau2017prochlo,allen2019algorithmic,censusbureau}).

\begin{table}[t]
\ifCCS
\vspace{5pt}
\fi
\centering
\caption{
Summary of available implementations of differentially private (DP) systems.
We identify systems that run DP code in a Trusted Execution Environment (TEE),
answer single queries only, \ie without global privacy budget (None),
maintain the privacy budget stored in runtime memory only (Runtime) 
and store the budget in a file in persistent storage (Persistent).
\dpsystem is the only system that uses TEEs (Intel SGX) and guarantees state continuity for the privacy budget
stored in persistent storage.
}

\setlength\tabcolsep{3.7pt}
\begin{tabular}{@{}l|p{1.3cm}|c||p{1.15cm}|p{1.9cm}}
\hline
\multirow{3}{*}{DP Systems}&\multicolumn{2}{c||}{\textbf{Components}}&\multicolumn{2}{c}{\textbf{Properties}}\\
\cline{2-5}
&\centering{Budget Storage} 
&{\rotatebox[origin=c]{90}{\hspace{-6mm}TEE}}
& \multirow{2}{*}{Liveness} & \multirow{2}{*}{Confidentiality} \\
\hline
Uber SQL~\cite{ubersql}		&		 \multirow{2}{*}{\rotatebox[origin=c]{0}{~~~None~~}}	&&		&		\\
PROCHLO~\cite{bittau2017prochlo}	& 	&	\cmark	&		&	 \multicolumn{1}{c}{\cmark}	\\
\hline
  OpenDP\footnote~\cite{opendp-lib,opendp}	&	\multirow{6}{*}{\rotatebox[origin=c]{0}{Runtime}}	&		&		&		\\
Tumult Anal.\footnote[1]~\cite{tumult}		&	 	&	&		&		\\
PINQ\footnote[1]~\cite{pinq}			&		& 	&		&		\\
GoogleDP~\cite{googledp}		&		& 	&		&		\\
Allen~\textit{et al.}~\cite{allen2019algorithmic} &  & \cmark& & \multicolumn{1}{c}{\cmark}	  \\
DuetSGX~\cite{duetsgx}		& 	& \cmark	&		&	 \multicolumn{1}{c}{\cmark}	\\
\hline
PySyft~\cite{Ziller2021}		&		 \multirow{2}{*}{\rotatebox[origin=c]{0}{Persistent}} 	&&  \multicolumn{1}{c|}{\cmark}	& 		\\
\textbf{\dpsystem}			&  &  \cmark	&  \multicolumn{1}{c|}{\cmark}	& \multicolumn{1}{c}{\cmark}	\\
\hline
\end{tabular}

\label{tab:budget-impl}
\end{table}

Many of these deployments use the \emph{Local DP} model where the organization collecting the data is not trusted and each user
applies DP algorithms locally to their own data before sending it 
to the organization (\eg Apple, Microsoft and Google use local DP for collecting telemetry results).
On the other hand, \emph{Global DP}
assumes that the \emph{trusted curator} collects raw data and answers queries of an untrusted analyst by running DP algorithms.
Though global DP provides much better utility than the local model~\cite{bittau2017prochlo}, 
deploying it in practice is difficult since the organization 
who is collecting the data is often also the one analyzing it (\eg tech companies build services from the insights they draw from their users' data).
One could envision creating a server that hosts the data and provides a DP query interface to it.
However, where would this server reside? The problem of trust persists if it is hosted on the premises of the organization.
 If, instead, a cloud service was used, the cloud provider would need to be trusted with the raw data --- likely not acceptable
for the data that required DP protection in the first place.

Running DP code on encrypted data in a Trusted Execution Environment (TEE)~\cite{SGX} can
prevent the organization or cloud provider from seeing the data in the clear. However,
we show that running existing DP systems in TEEs is not sufficient to provide DP guarantees and
protect data against reconstruction attacks.
These attacks stem from an essential, but often overlooked, part of a DP system: the privacy budget.
Since the privacy budget is used to constrain the number of queries the system can answer while providing DP guarantees, 
any tampering with the budget can lead to information leakage. 
For example, by rolling back the budget or answering multiple queries from the same budget,
the attacker gets access to many more outputs than what is permitted by the budget.
Even though each query output will have DP guarantees, 
the results can be averaged to cancel out the noise.
Indeed it is known that even if  noise is applied to the queries' output, 
the more queries are answered the more information is released about a dataset~\cite{dinur2003revealing}.

\footnotetext{Description of these systems suggests that the budget is maintained only during a DP session  and does not explicitly mention persisting the budget.}

In this paper we address the problem of designing a secure global DP system in an untrusted setting.
We first identify that a secure system
needs to provide two key properties:
\emph{liveness}, requiring that after a crash it can recover with the latest privacy budget; and
\emph{DP confidentiality} (or confidentiality for short), meaning that it leaks no more
information than one in which the curator is trustworthy---even in the
face of collusion between the cloud provider hosting the DP system, the
network, and the data analyst who submits queries to the curator.
Since DP mechanisms are randomized, we
cannot insist on the DP system producing the same answers as
those produced by a trusted curator. To this end,
we propose a notion of equivalent transcripts over the distributions of answers.
We then design a system, \dpsystem, that guarantees these properties even when an attacker
controls the persistent storage.

\parhead{Existing DP Systems.}
\cref{tab:budget-impl} shows that existing DP systems either do not provide liveness (\ie cannot recover privacy budget after a crash)
or cannot guarantee DP confidentiality in an untrusted system.
For example, DP systems that maintain the budget only in runtime are prone to losing it if a system crashes (even if the crash is benign).
Indeed, only one system stores the budget in persistent memory.
However, none of the systems protect against replay, forking or rollback attacks against the budget if the system were restarted. 
We use two (insecure) instantiations below to highlight the difficulties in securely extending
previous DP systems:

\emph{(Insecure) Instantiation 1:}
Consider a global DP system that runs in a TEE~\cite{allen2019algorithmic,duetsgx}.
In this system to protect data from a cloud service provider, 
the data owner submits raw data, encrypted, which is then decrypted 
in the TEE. To protect against the analyst, data is
processed with a DP algorithm by the TEE, ensuring that only DP-protected
output ever leaves the TEE in the clear. 
To accommodate multiple queries, budget is uploaded in the TEE's runtime memory. 
Since the budget is stored and updated within the TEE, the budget integrity is protected. Once the budget is exceeded the TEE shuts down.
However, if the TEE crashes, one will not know how many queries were
answered and if budget was exceeded. Hence, such a system cannot provide DP confidentiality after a crash.
This instantiation is inspired by systems that support runtime budget in \cref{tab:budget-impl} (\eg GoogleDP~\cite{googledp}).

\emph{(Insecure) Instantiation 2:}
To be able to restart and provide DP guarantees, the privacy budget should be stored in persistent storage.
However, such storage is provided by an untrusted cloud provider and can be easily tampered with
(validated empirically in \cref{sec:sgx-attack}).
Signatures on the stored budget can protect from an attacker setting their own budget value; however, they do not protect against attackers supplying
previously stored budget states to
a restarted TEE
(\ie rollback attacks) or creating multiple TEEs with the same budget (\ie fork attacks).
Such budget attacks invalidate DP guarantees on the queries answered after a restart.
(Instantiation~2 is inspired by PySyft~\cite{Ziller2021} which is the only system that persists its budget in a database; albeit without protection against
a malicious cloud provider or budget attacks.)

\parhead{\dpsystem.}
Since none of the existing systems can protect against budget attacks in case of a system crash,
we propose \dpsystem 
that provides both liveness and confidentiality.
Our system relies on a carefully designed protocol around a state continuity module (SCM), 
which tracks budget updates, and a TEE, which runs the DP code.
Though there is a significant amount of existing work on each of these components, integrating them with each other
is not trivial and can lead to an insecure implementation.
First,
an incorrect order in which the budget is updated and the answer to a query is returned
can reveal answers to more queries than what the budget should allow.
Second, DP algorithms are randomized; if an answer to a previously answered query is recomputed, it is very likely to return a different result --- facilitating a differencing attack to recover the sensitive output.
Finally, to guarantee liveness, the (untrusted) environment should be able to (re)start TEEs.
The latter can lead to multiple
TEEs answering queries in parallel --- enabling a fork attack as a~result.

To address the above challenges, we make several design choices that contribute to the security of the overall system, 
including what needs to go into the persisted state, and the order in which
the budget is updated and the query is replied.
For example, \dpsystem records the answer {to the last query} 
and the updated budget in untrusted storage and SCM, before the answer is sent to the analyst.
We prove that \dpsystem guarantees liveness and DP confidentiality, while relying on TEE security and
the existence of minimal state continuity functionality (\eg TEE-based distributed systems for state maintenance 
such as Nimble~\cite{angel2023nimble} support this).
Our system provides confidentiality guarantees even if the data analyst colludes with the storage provider, 
intercepts the network between SCM and TEE,
and could crash and restart the~TEEs.

In summary, our contributions are:

\begin{itemize}[nosep,leftmargin=*]
\item \emph{Problem formulation.} We highlight that security of a DP system crucially relies on faithfully maintaining the privacy budget. We show that existing approaches do not protect the budget and, hence, are vulnerable to attacks that invalidate their DP guarantees (even if the code runs in a~TEE).
\item \emph{DP security requirements.} We formulate the threat model and properties of a secure instantiation of global DP model when executed by an untrusted system provider.
\item \emph{Protocol.} We propose \dpsystem, first DP system that explicitly considers protection of budget's continuity between queries, thus 
  providing DP guarantees akin to those of~a~trusted~curator.
\item \emph{Implementation.} We implement \dpsystem using Intel SGX enclaves to run DP code 
and a state continuity module (provided by a distributed network of TEEs~\cite{niu2022narrator}).
Compared to an insecure implementation, our system provides 1.1--3.2$\times$ query overhead.
We observe that the {relative} overhead of \dpsystem is lower when the DP system processes complex queries.
\end{itemize}

\section{Background}
\label{sec:background}

\paragraph{Global Differential Privacy.} 
Consider a space of possible datasets $\mathcal{D}$
with a binary \textit{neighboring} relation
--- two datasets~$\cd, \cd' \in \mathcal{D}$ are neighboring if they differ in a single record 
--- and consider an arbitrary space of queries $\mathcal{Q}$ against datasets, 
and possible query responses as outputs in $\mathcal{R}$. 
We abstract the particulars of differentially-private mechanisms 
and consider general randomized mechanisms $\mech: \mathcal{D}\times\mathcal{Q} \rightarrow \mathcal{R}$, 
that release some random answer $\answer$ when running the mechanism $\mech(\cd)$ on a dataset and query~$\query$. 
We consider $\mech$ to be $(\epsilon,\delta)$-\textit{differentially private}~\cite{dwork2006calibrating} 
for $\epsilon>0$ and $\delta\in[0,1)$ if given any two neighboring $\cd, \cd' \in \mathcal{D}$ 
and any subset of outputs $R \subseteq \mathcal{R}$
it holds that $\Pr[\mech(\cd) \in R] \le \exp(\epsilon) \Pr[\mech(\cd') \in R] + \delta$.
This means that except for unlikely events --- those with negligible probabilities --- the probability of observing outputs from $\mech(\cd)$ or $\mech(\cd')$ 
are always close and bounded by $\exp(\epsilon)$. (Note that for small $\epsilon$ budget, $\exp(\epsilon)\approx 1+\epsilon$.) 
For $\epsilon$-DP, one needs to ensure that $\Pr[\mech(\cd) \in R] \le \exp(\epsilon) \Pr[\mech(\cd') \in R]$ holds.
Differential privacy therefore guarantees that an attacker has limited ability 
to determine whether $\cd$ or $\cd'$ is used as input to a mechanism, based on its outputs.

Suppose that $\epsilon_i$ is the privacy loss for answering query~$i$. 
Due to the (basic) composition property of DP, if responses to~$k$ 
such queries are released from the same dataset, the total privacy loss is at most $\sum_{i=1}^{k} \epsilon_i$ for that dataset.
A data owner may assign a dataset with a total privacy budget~$\budget$. 
Then maintaining the budget means disallowing any queries to this dataset that would exceed
the budget, that is not answering a $k$th query if $\sum_{i=1}^{k} \epsilon_i > \budget$.
To support this behavior, we define $\spec(\query,\budget)$ that returns a tuple $(\epsilon,\delta)$ that determines $\query$'s
expended privacy loss and a
DP mechanism $\mech$ for answering query $\query$ 
(\eg depending on the query, it may return a Laplace or a Gaussian mechanism) given 
the current budget (\eg if the budget is low, the specification may suggest using a smaller $\epsilon$ 
to expand the budget at the expense of less accurate answers, see Section~\ref{sec:discussion} for budget management strategies). 
Then, given dataset~$\cd$, $\mech$ computes the differentially-private answer~$\answer$.
While $\answer$ is computed using the dataset~$\cd$ and secret random coins, privacy parameters $\epsilon, 
\delta$ are computed deterministically and independently of~$\cd$. 
From now on, we omit $\delta$ and consider a system that supports $\epsilon$-DP where $\spec$ returns $\mech$ and $\epsilon$ only.
For $(\epsilon,\delta)$-DP systems the budget would be set for both~$\epsilon$ and~$\delta$ and both will be updated after each query.
 
In this study, we show that budget $\budget$ can be attacked (\eg rolled back or reused) 
and propose a mitigation to faithfully maintain the state continuity of the budget for a dataset in a DP system. 
While past work on composition considers the accumulation 
of privacy loss~\cite{dwork2010boosting,pmlr-v37-kairouz15,dwork2014algorithmic}, 
the focus of such work is distinct to our own: 
there the goal is to more accurately analyze theoretical privacy loss across multiple queries 
beyond the (worst-case) bound of basic composition. 
Existing DP systems implement these composition properties and conveniently track privacy loss, 
however they offer no state continuity guarantees.
We seek state continuity of this accumulated privacy loss in system implementations 
even in the face of attackers attempting to tamper with the budget.

\paragraph{Trusted Execution Environment.}
Trusted Execution Environments (TEEs) allow the sensitive part of an application 
to run in an isolated and memory protected area (\eg an Intel SGX enclave~\cite{SGX}).
Hence, one can deploy applications with TEEs in a potentially untrusted cloud environment.
TEEs provide several properties: 
(1) The confidentiality and integrity of an application.
When an application is running, 
its memory is held in an isolated and protected area that no privileged processes/users can read nor tamper with.
When an application transitions to a state of rest, it may choose to seal
its state (encrypted with authenticated encryption) and recover it later.
(2) The remote attestation feature guarantees that one can confirm that
the code running in a TEE is the expected code and running on the legitimate hardware~\cite{Costan2016IntelSE}.
Though our experiments rely on Intel SGX, our protocol can be instantiated with any other
trusted module that provides the above guarantees.

\paragraph{State Continuity.}
\label{sec:scm}
TEEs, such as those provided by Intel SGX, do not provide persistent state. If a
TEE crashes or is maliciously restarted, its runtime state is
lost. Hence, applications running in TEEs must explicitly address this.
Unfortunately sealing described above can only ensure that the state content is protected
and a TEE can detect if it was tampered with. However, it cannot detect if the state it is unsealing
is the most recent or whether there is another TEE that is already running from the same state.
These attacks are possible in the case of a compromised storage service or operating system.
We observe that for a TEE running DP code this means that one can either reset the budget that has been stored in a sealed state --- mounting a rollback attack, 
or start more than one enclave from the same state --- mounting a fork attack. To this end, 
one needs to rely on external methods for maintaining state continuity~\cite{jangid2021towards}.

\begin{algorithm}[t]
\caption{Functionality provided by $\SCM$.
All functions are executed atomically.}
\begin{algorithmic}
\State $(\sk_\scm, \vk_\scm)$ \Comment{SCM's signing and verification key}
\State $\id_\scm$ \Comment{State id stored at SCM}

\State State $\curstate$ \Comment{Current state stored at SCM}

\\

\Function{$\InitState$}{$\nonce, \newstate$} \Comment{Set initial state}
   \State $\curstate \gets \newstate$
   \State $\id_\scm \gets 0$
    \State \Return $\Sign(\sk_\scm, \id_\scm, \curstate, \nonce)$ 
\EndFunction

\\

\Function{$\ReadState$}{$\nonce$} 
   \State $\sigma_\scm \gets \Sign(\sk_\scm, \id_\scm, \curstate, \nonce)$
   \State \Return $\sigma_\scm, \id_\scm, \curstate$ 
\EndFunction

\\

\Function{$\UpdateState$}{$\nonce, \id, \newstate$}
\If{$\id = \id_\scm + 1$}
   \State $\curstate \gets \newstate$
   \State $\id_\scm \gets \id$
   \State   $r \gets \mathsf{Ack}$
\Else
   \State $r \gets \mathsf{Error}$
\EndIf
 \State \Return  $\Sign(\sk_\scm, r, \nonce)$, $r$
\EndFunction

\end{algorithmic}
\label{alg:scm}
\end{algorithm}    

State continuity can be obtained using hardware primitives such as monotonic counters,
NVRAM or TPM~\cite{strackx2016ariadne,parno2011memoir, brandenburger2017rollback}.
However, such solutions present one point of failure.
An alternative is to store a state replicated across multiple nodes (e.g., via a distributed ledger).
To guarantee fault-tolerance in a malicious setting
they assume that a majority of the nodes is alive and honest (e.g., implemented using TEEs as 
in~ROTE~\cite{matetic2017rote}, Narrator\cite{niu2022narrator}, Nimble~\cite{angel2023nimble} and Confidential Consortium Framework~\cite{russinovich2019ccf}).

Our protocol relies on an SCM abstraction and is independent of the exact implementation.
We abstract the required functionality of state continuity module~(SCM) in~\cref{alg:scm}.
It closely resembles the one presented in Nimble~\cite{angel2023nimble} while explicitly stating the signing and state update operations of the service.
The SCM maintains $\curstate$ and the counter of the number of times it was updated,~$\id_\scm$.
The size of the state supported by SCM depends on its implementation. 
For example, if it is implemented with a TPM, the state is small (e.g., a hash of the application state) to fit in the TPM's available non-volatile memory.
For our protocol, we require the SCM to sign its replies to ensure that
states sent to TEEs are from a valid SCM.
Queries sent to the SCM are augmented with $\nonce$, requiring SCM's signature to include it.
The nonce ensures that the reply from the SCM is current and prevents replay attacks by a malicious network.
SCM updates the state only if it was requested to update it with a state associated with the next counter value (\ie $\id_\scm +1$).

Our protocol will leverage SCM's capabilities to build a secure DP system. 
In~\cref{sec:overview}, we discuss several naive (but insecure) designs of integrating SCM and DP.
To this end, the proposed system, \dpsystem, will need to carefully decide on what DP components to persist in a state
and saved at the SCM, in which order to update the state locally, at the SCM, and when to send the
reply to the analyst. In~\cref{sec:analysis}, we prove these design choices lead to a secure system.

\section{Setting and Threat Model}

In this section we outline the threat model when instantiating a trusted DP curator in an adversarial environment.
We then show that existing implementations set in this threat model are vulnerable to budget attacks
and show how one can reconstruct the statistics that DP aims to protect.

\subsection{Problem Setting}
We consider a setting where 
a \emph{data owner} has a static dataset~$\cd$.
She wishes to setup a cloud service 
that would run on her behalf and answer \emph{data analyst's} queries on this dataset.
The data owner wants analyst's queries to be answered with differential privacy guarantees 
as long as they do not exceed a certain budget $\budget$ (set by the owner). 
Moreover, the budget should be reduced by~$\epsilon$ associated with each answered query.

In addition to the data owner and an analyst, our setting consists of a \emph{server}.
The data owner uploads her {(encrypted)} data  $\cd$ and associated budget $\budget$ to a server. 
The server offers storage where the data owner can store their data.
It also offers protected computation resources, abstracted as a \emph{trusted execution environment, TEE} 
(\eg implemented using Intel SGX) and regular computation resources.
TEE can access server's storage and run DP code.
Our solution will also make use of a \emph{state continuity module (SCM)} 
that has the basic functionality as described in \cref{alg:scm}.

\subsection{Threat Model}\label{sec:tmodel}

\parhead{Trusted Components.}
We assume the data owner, TEE and SCM to be trusted.
That is, TEE provides isolation to the code and data in its runtime memory:
its code and data cannot be tampered nor observed.
However, once it crashes, it no longer has access
to its private memory. 
{We assume that a newly instantiated TEE has access to its own pseudo-random source,
independent of other instantiated TEEs. That is, even if TEEs were to use the same keys for encryption,
they would use independent sources of randomness to provide semantic security.}
Though there is much work on side-channel attacks
and protection guarantees of some instantiations of TEEs (\eg for DP algorithms~\cite{allen2019algorithmic}),
we consider side channels~\cite{gotzfried2017cache,brasser2017software} to be out of scope for this paper.
The code of~DP~$\mech$ is trusted and will answer queries with DP guarantees.
We refer to a TEE as \emph{valid} if it is running the code of a data owner 
(\eg it has passed a remote attestation check in case of Intel~SGX).

{SCM is trusted to provide state continuity and follow the functionality in \cref{alg:scm}.
That is, it will only update the state if it corresponds to the next
counter value from the one it stores locally. Note that SCM does not check who it receives
state updates from as long as they supply the correct counter value.
Our system will be able to detect if the state was updated by a malicious party.}

\parhead{Attacker.}
We assume that a data analyst is not trusted as otherwise the data owner would not use DP mechanisms to answer their queries.
The server is not trusted and may collude with the data analyst. 
For example, the TEE host could be curious about their customers' activities or could be compromised by malware (a common threat model of a TEE).
To this end, we consider an attacker who controls the data analyst, the server, and the network.
Though the traffic between the analyst and the TEE can be protected
by, for example, TLS, this will not protect against the above attacker.

\textbf{Attacker's goal} is to learn more information about the dataset than what is released to it via the DP querying interface.
For example, it will try to issue more queries than it can,
tamper with data that TEE stores in the external storage (\ie the persistent storage provided by the server), 
give the TEE old copies of the data it stored 
(\eg provide older copies of the budget if budget was stored outside of TEE's runtime memory,
thereby rolling back the budget),
snoop on the messages TEE sends to SCM, send messages to SCM on TEE's behalf 
and vice versa (\eg by sending its own messages or replaying TEE's and SCM's messages).
Given that the attacker controls the server, including the operating system,
it can arbitrarily crash and start TEEs, thereby aborting and restarting processes running in these TEEs.
Alternatively, it can also initialize TEEs with copies of the same persistent storage 
(\eg creating two TEEs with the same budget and sending them same or different queries, thereby executing a fork attack).
In \cref{sec:sgx-attack}, we show that these capabilities can be easily instantiated in practice to invalidate DP guarantees.
\subsection{Budget Attack}\label{sec:sgx-attack}

We now empirically show that Instantiation~2 described in the introduction
is vulnerable to a budget attack in the threat model outlined in the previous section, even when DP code runs in Intel SGX.
In this instantiation, a data owner uploads encrypted dataset to the cloud along with a file containing initial encrypted budget.
DP code is running within a TEE and the budget is protected with authenticated encryption
and saved to persistent storage.
A TEE reads the budget from a file,
answers a data analyst query, and writes the updated authentically encrypted budget back to the file.
As a TEE we use Intel Xeon E-2288G CPU, 
which supports SGX functionality.

\paragraph{Dataset and Query Model.}
We use 1000 records from the PUMS dataset~\cite{pums} in the attack, 
and the discrete Gaussian mechanism~\cite{canonne2020discrete} to protect query output.
We simulate a setting where data analyst
queries the PUMS dataset using count queries on the age and income columns
(\eg count \#records of a certain age).
We assume that a data owner sets the overall budget $\budget=10$ for the dataset.

\paragraph{Attack.}
Based on the above threat model,
 an attacker { can control the operating system of the server}.
Hence, it has the privilege to kill a process running in a TEE using the system's kill command
and initialize a new process. Since the budget is stored in a filesystem controlled by the attacker,
upon a file request from a TEE,
it can provide files of its choosing.
To this end, we implement an attacker who kills a process running in a TEE whenever
the remaining budget is exceeded.
The attacker then creates a new TEE.
When TEE requests a budget file, the attacker provides it with
the initial version of the encrypted and signed budget file that has the full budget
instead of the latest file with no~budget~left.
Since SGX cannot detect rollback attacks (discussed in~\cref{sec:tmodel}), it cannot
detect that it is being queried more times than what the budget~allows~for.

The attacker repeats the same query $N_R$ times to average out noise added by the DP mechanism.
To evaluate the effectiveness of the attack we use the relative mean squared error ($\RMSE$):
$\RMSE = {\sum(S_i-S'_i)^2}/{\sum S_i^2}$,
where $S_i$ is the value of the true answer to the~$i$th query, 
and $S'_i$ is the guess the attacker makes.

\paragraph{Attack Results.}
We plot the relationship between $\RMSE$
 and $N_R$ ranging from 1 to 10, for $\epsilon$ values 0.5 and 1, in~\cref{fig:reconstruct_pums_n}.
{As expected,} $\RMSE$ decreases as the noise cancels out with more queries.
As a baseline we use the error due to the Gaussian noise added by DP,
that is, when $N_R$ is 1.
For example, the error is around 0.25 and 0.1 for income for $\epsilon$ of~0.5 and~1, respectively.
Recall that smaller $\epsilon$
results in larger noise, hence, the error is higher for $\epsilon=0.5$.

\begin{figure}[t]
  \includegraphics[width=0.37\textwidth]{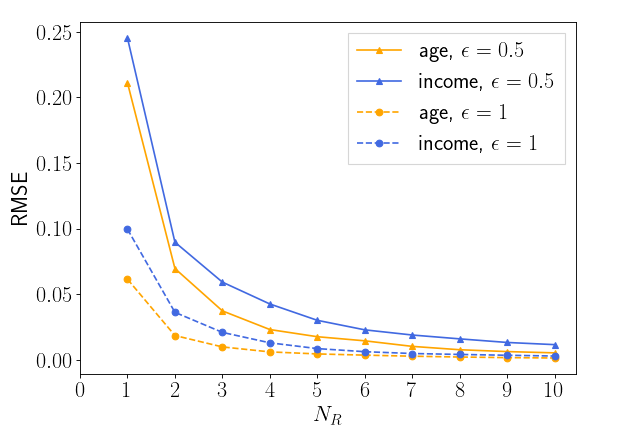}
  \centering
  \caption{
    The error (RMSE) of attacker's guess for the
    count query answer
    under rollback attack: the budget is reset when it is exceeded
    and each query is repeated
    $N_R$ times to average out the noise.
  }
  \label{fig:reconstruct_pums_n}
\end{figure}
\section{Design and security Goals}
Given the above threat model we wish to design a DP system that, at a high level, 
behaves as a trusted third party executing DP code would behave when set in an untrusted environment.
In particular, the system should answer analyst's queries
with DP guarantees and within privacy budget constraints, while also being able to recover from a crash.
We propose \emph{liveness}  and \emph{DP confidentiality} definitions to capture these properties.

\subsection{Liveness}
We define liveness as the ability of the system to recover in case of a crash.
For DP systems, that means being able to recover its latest privacy budget
and, if the budget is not exceeded, to continue answering analyst's queries.
Note that the attacker can always delete data stored by a TEE externally, 
refuse to start a TEE or block
any messages between TEE and the rest of the system.
Hence, like other distributed systems, we can guarantee liveness only when a server is honest~\cite{angel2023nimble,parno2011memoir,strackx2016ariadne}.
We discuss this further in Section~\ref{sec:discussion}.

\begin{definition}\label{def:crash-resilient}
    A DP system provides \emph{liveness} if it can recover its latest privacy budget
    after a crash.
\end{definition}

\subsection{DP Confidentiality}
This property ensures that the attacker learns about the dataset~$\cd$ as much
as it would when interacting with a trusted DP curator who faithfully maintains a privacy budget.
That is, every query is answered using a DP mechanism, the budget is correctly adjusted after
each answered query and queries that would exceed the budget
return answer~$\bot$. 

Formally stating DP confidentiality is not straightforward for two reasons. Non-trivial DP mechanisms
are randomized and the order in which queries are answered may alter their answer 
(\eg budget may expire earlier if more expensive queries are answered first).
Hence, even if a system provides DP guarantees, the answers may be different 
from another DP system due to different randomness used or the order
in which queries are answered.
We overcome this challenge by defining a transcript that captures 
each requested query, including its answer and the remaining budget.
We then say that a system provides DP confidentiality, if it produces a transcript that is equivalent to a transcript 
produced by the ideal system (\ie trusted curator). 

\paragraph{Ideal DP System.}
\label{sec:idealdp}
Our system aims to resemble the following ideal setting.
The data owner is the only one that has access to the dataset and the privacy budget.
It sets the budget $\bm{\budget}$ and answers data analyst's queries in a sequential order
as long as there is enough budget to answer the queries. If the budget is not enough, it sends
$\bot$ as a reply. Given a query~$\query$ and the remaining budget, 
$\spec$ returns the DP mechanism~$\mech$ and~$\epsilon$ required to answer~$\query$. The query is answered using this mechanism if there is enough budget.
The data analyst learns only the DP answers to the queries.

We present this ideal DP functionality in~\cref{alg:idealdp}. It takes
a sequence of queries $Q$ and processes them in the order given in the sequence.
It then constructs a transcript  $T_\mathsf{ideal}$ where each entry is 
of the form $(\query, \answer, \budget)$ where $\answer$ is the answer to the query 
and $\budget$ is the updated budget after the query.

\begin{algorithm}[b]
\caption{Ideal DP Functionality for answering a sequence of queries $Q$ on dataset $\cd$ given total budget $\bm{\budget}$.}
\begin{algorithmic}
\Function{$\mathsf{IdealDP}(\cd, \bm{\budget}, Q)$}{}
\State $i \gets 1$ \Comment{Query counter}
\State $\budget \gets \bm{\budget}$
\For {$\forall \query \in Q$} \Comment{Process queries in sequence order}
\State $\epsilon, \mech \gets \spec(\query, \budget)$
\If{$\budget - \epsilon \ge 0$}
\State $\answer \gets \mech(\cd)$
\State $\budget \gets \budget - \epsilon$ 
\Else
\State $\answer \gets \bot$ \Comment{Query exceeds the budget}
\EndIf
\State $T_\mathsf{ideal}[i] \gets (\query, \answer, \budget)$
\State $i \gets i +1$
\EndFor 
\EndFunction
\end{algorithmic}
\label{alg:idealdp}
\end{algorithm}    

Note that differentially-private algorithm $\mech$ is a randomized algorithm.
Hence, even transcripts of two ideal functionalities
on the same set of queries are likely to produce different answers (but same budget updates).
To this end we say that two transcripts~$T_1$ and~$T_2 $ are equivalent as long as the answers
they produce are both~$\bot$ (\ie there was not enough budget to answer this query) or they are both elements of the same distribution $\mech(\cd)$.
\begin{definition}[Equivalent transcripts]
For a dataset $\cd$,
transcripts~$T_1$ and~$T_2 $ are equivalent, denoted as $T_1 \approx T_2$, 
if they have the same number of entries and for every entry $i$:
\begin{itemize}
\item
$T_1[i].\query =  T_2[i].\query$, $T_1[i].\budget =  T_2[i].\budget$ and
\item Either $T_1[i].\answer$, $T_2[i].\answer$ are both $\bot$ or from $\mech(\cd)$.
\end{itemize}
\end{definition}

We then say that a system provides DP confidentiality for a sequence of queries, 
if it updates the global budget and answers queries as a DP system running on a single trusted machine would. 
\begin{definition}[DP Confidentiality]
\label{def:securedp}
Let $Q$ be a sequence of all queries answered by a DP system
and $T$ be the corresponding transcript.
DP system provides DP Confidentiality if there is an order on queries in $Q$, $\pi$, such that
the generated transcript~$T$ is equivalent to the transcript,
$T_\emph{ideal}$, generated by the ideal DP system on $\pi(Q)$  (\ie $T_\emph{ideal} \approx T$).
\end{definition}

Note that in the Ideal DP system the budget cannot be tampered with. As a result, it is protected from rollback attacks.
Moreover, since there is only one copy of the budget and all queries are answered against it synchronously, the budget cannot be forked.
Hence, a DP system that provides DP confidentiality (i.e., behaves the same as the Ideal system) also protects against budget attacks.
\section{\dpsystem Protocol}
\label{sec:proto}

\begin{figure*}[t]
\includegraphics[width=0.9\textwidth]{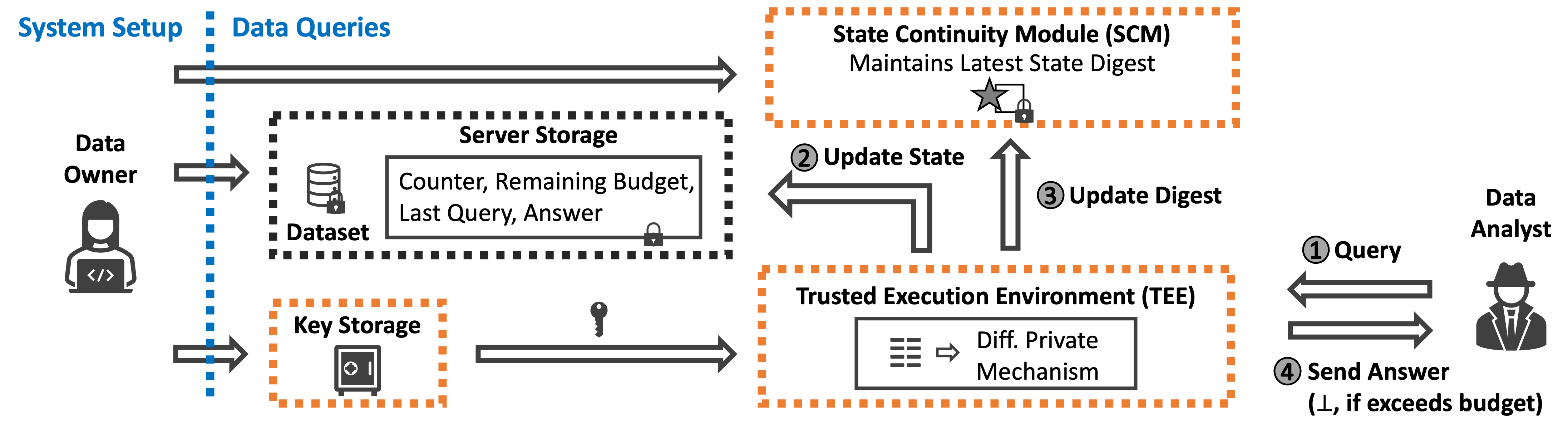}
\centering
\caption{
The overview of \dpsystem where a TEE answers analyst's queries using DP mechanisms on data owner's behalf.
During System Setup the data owner uploads the encrypted dataset and state containing the initial privacy budget to the server.
It uploads the hash of the encrypted state to a State Continuity Module (SCM) 
and secret keys to the key storage. 
When TEE is initialized (not depicted in the illustration), it loads the keys, the data and the state. 
It checks that the state is most recent by comparing it against the digest maintained by the SCM. 
When the analyst sends the query (Step 1), TEE computes a DP answer and stores the query counter, 
the updated budget, the query and the answer in persistent storage (Step 2).
It then tries to update the SCM with this new state (Step 3). 
If the update is successful (\eg no other TEE has updated the state in the meanwhile), 
the TEE sends the answer to the analyst. The data owner and the components in orange are assumed to be trusted.
} 
\label{fig:dpsystem}
\end{figure*}

Existing DP systems can guarantee DP confidentiality only if the corresponding system does not crash (\cref{tab:budget-impl}).
If they were to provide liveness by recovering from a crash, 
they would be vulnerable to budget attacks as described in the introduction and shown in~\cref{sec:sgx-attack}.
To this end, we propose \dpsystem (illustrated in \cref{fig:dpsystem}), a secure
protocol for instantiating global DP model in an untrusted environment.

At a high level \dpsystem achieves liveness by storing the current privacy budget in a persistent but insecure storage.
It achieves DP confidentiality by processing decrypted data only within a TEE.
To ensure that the budget it reads from insecure storage is most recent,
it uses a state continuity module (SCM) to maintain a digest of the state that includes the budget.

\subsection{Overview}
\label{sec:overview}
There are two challenges in integrating an SCM in a DP system. The first one is what components
to persist in a state and the second one is in which order to update this state and answer the query.
Below we describe the intricacies behind these choices.

DP budget is a natural candidate to be persisted in a state.
However, answering queries before the state digest is updated at the SCM,
would lead to a system vulnerable to rollback attacks.
Hence, the query must be answered after the budget is updated.
However, a TEE may crash after updating the budget but before sending the answer to the analyst.
Once the system restarts, the query would remain unanswered even though it has used up the budget.

To avoid the depletion of a budget above, one may suggest to augment the state with the
last query. Then, when the TEE restarts, the answer to the recorded query could be recomputed without
using the budget again.
However, this opens a possibility of a differencing attack since a previously computed answer may have been sent before TEE crashed.
Since DP mechanisms are randomized, it is likely the attacker will get a different answer 
than before. It~can then crash the TEE after every answer 
(i.e.,~mimicking the attack in~\cref{sec:sgx-attack} against a system that did not use SCM).
To mitigate this, \dpsystem records the budget, the query and the answer in its state and re-sends the answer instead of recomputing~it.

The last design choice arises from the order in which the state and its digest need to be updated at the untrusted local storage and the SCM, respectively.
If the SCM is updated before the local state is updated and TEE crashes, no other TEE will be able to restart as the non-hashed version of a state is lost. However, if the local state is updated first and the TEE crashes, a restarted TEE will be out of sync with SCM. \dpsystem adopts the latter option while still being secure. In particular, it allows a TEE to update the SCM using the (authenticated) version of the state stored at the server as long as it is at most one version higher than SCM.
To ensure that the attacker does not start multiple TEEs with the same budget and launches a fork attack, no query (including the one recorded in the state) is answered until the SCM is in-sync
with the local state. Here, SCM guarantees that only one of the TEEs will be able to update the state,
while all others will abort.

\subsection{System State}
\label{sec:state}

\label{sec:state-component}
\dpsystem's state consists of
a tuple $s = \query\|\answer\|\budget$ and
a query counter/identifier $\id$ where:
\begin{itemize}
\item $\query$ is the most recently submitted query from the data analyst, including all the arguments required for a query (\eg query type). During initialization $\query$ is set to $\bot$.
\item $\answer$ is the answer to $\query$. It is set to $\bot$ at the initialization or if there was not enough budget to answer the query.
\item $\budget$ is the remaining budget, initially set to the total budget~$\bm{\budget}$ allocated by the data owner for the dataset~$\cd$.
\item $\id$ is the counter of all answered queries, including those that returned $\bot$ due to budget constraints. This is also a state identifier as the state is updated for every query.
\end{itemize}

System digest $h$ is set to $\hash(\encs)$ where $\hash$ is a collision-resistant hash function
and $\encs$ is the encrypted version of the state.
We use hash of encrypted state since the state contains secret information $\answer$. 
If the hash is received before the update to the state has been recorded, 
one could use a dictionary-like attack to guess the answer.

\subsection{Protocol Entities}

\dpsystem consists of the following entities (see~\cref{fig:dpsystem}).
\begin{itemize}[nosep]
    \item (Trusted) Data owner of the static dataset $\cd$.
    \item (Untrusted) Data analyst that submits queries about~$\cd$.
    \item Trusted Execution Environment (TEE) that answers queries using a DP mechanism.
    \item State continuity module (SCM) (defined in~\cref{sec:scm}) that keeps track of the most recent state digest~$h$.
    \item (Untrusted) server that stores static encrypted and integrity protected dataset~$\enccd$. 
    It also stores the encrypted system state~$\encs$ that is updated by a TEE.
\end{itemize}

The TEE interacts with server storage ($\ST$) via $\ST.\Load(\text{tag})$ and $\ST.\Store(\text{tag}, \text{object})$
calls that load an object stored with the provided \text{tag} and store an object under the specified tag, respectively.

The TEE interacts with the state continuity module ($\SCM$) via 
$\SCM.\UpdateState(\nonce, \id, \newstate)$ and $\SCM.\ReadState(\nonce)$ that are defined in~\cref{sec:scm}.

\subsection{System Setup}
The data owner sets up the system using dataset $\cd$ and initial budget~$\bm{\budget}$ as follows
{(the pseudocode is provided in \cref{alg:setup})}.
It sets initial state~$s_0 = \bot\|\bot \|\bm{\budget}$ with identifier 0.
Since the data owner provides the dataset that needs to be protected,
she
generates two symmetric keys~$\datakey$ and~$\statekey$ of 
an Authenticated Encryption scheme for encrypting the dataset~$\cd$ 
and state $s_0$ respectively (\cref{line:key1} and~\ref{line:key2} in \cref{alg:setup}).
The data owner also generates a signing and verification key pair~$\sk_s$ and~$\vk_s$ that TEEs will use to verify that state digests are authentic and have not been tampered with.

The data owner sends a signature $\sigma \gets \Sign(\sk_s,(0,h_0))$ 
over a hash of the encrypted state $h_0 \gets \hash(\encs_0)$ and a random~$\nonce$ to the SCM
via $\InitState(\nonce, (\sigma_{h_0}, h_0))$.
In return it obtains a signed reply from SCM. To prevent a person-in-the-middle or replay attacks from the SCM,
the data owner checks that the correct state hash was loaded on SCM and 
that it is communicating with the SCM it trusts (using $\vk_\scm$) and uses $\nonce$ to check
that the reply is current.
She then uploads the encrypted dataset $\enccd$, the encrypted state $\encs_0$, 
and a signature $\sigma_{h_0}$ to the untrusted storage at the server. 

The data owner needs to set up a way for a TEE, once initialized, 
to be able to access the keys $(\datakey, \statekey, \sk_s, \vk_s, \vk_\scm)$ so that it can act on her behalf
and answer data analyst's queries.
There are several ways to do it, including the following two.

We assume that the DP system code, including $\spec$ and DP mechanisms for answering the queries, 
is public and can be loaded by any TEE and is amenable to remote attestation.
The data owner could initialize a TEE with the DP systems' code and transfer the keys after checking
that the TEE is legitimate via remote attestation. The TEE could then seal and store them in untrusted storage
so that it can recover them without contacting the data owner if the crash occurs.
Alternatively, the data owner can use a key management service 
(\eg Azure Key Vault\footnote{\url{https://learn.microsoft.com/en-us/azure/key-vault/managed-hsm}})
and ensure that only TEEs with DP code as loaded by the data owner can access the keys.
To this end, we abstract the process by which the TEE obtains keys via
$\mathsf{GetKeys}(m)$ where $m$ is TEE's proof for remote attestation.

\begin{algorithm}[t]
\caption{{The data owner sets up \dpsystem to answer queries using her dataset $\cd$ with 
budget~$\bm\budget$. For state continuity, the system is setup to use $\SCM$ whose verification key is~$\vk_\scm$.}}
\begin{algorithmic}[1]
\Function{$\Setup$}{$\cd,\bm{\budget}, \vk_\scm$}
\State $\datakey\gets\genSymKey()$ \label{line:key1} \Comment{Encryption key for the dataset.}
\State $\statekey\gets\genSymKey()$ \label{line:key2} \Comment{Encryption key for the state.}
\State $\sk_s,\vk_s\gets\genAsymKey()$ \Comment{Asymmetric keys for signing and verification of the state.}

\State $s_0 \gets \bot\|\bot \|\bm{\budget}$ \Comment{Set initial state.}
\State $\encs_0 \gets \Encrypt(\statekey, s_0)$ and $h_0 \gets \hash(\encs_0)$
\State $\sigma_{h_0} \gets \Sign(\sk_s,(0,h_0))$ \Comment{Sign the hash of the encrypted initial state.}
\State Generate random $\nonce$
\State $\sigma_\scm \gets \SCM.\InitState(\nonce, (\sigma_{h_0}, h_0))$
\If{$\Verify(\vk_\scm, \sigma_\scm, 0, (\sigma_{h_0}, h_0), \nonce)$ fail} 
    \State Abort \Comment{State initialization at SCM failed.}
\EndIf
\State $\enccd\gets \Encrypt(\datakey, \cd)$
\State $\ST.\Store(\text{`data'}, \enccd)$ and  $\ST.\Store(\text{`state'}, (\sigma_{h_0}, \encs_0))$ \Comment{Upload the encrypted dataset and the encrypted state to the storage.}
\State \Return 
\EndFunction

\end{algorithmic}
\label{alg:setup}
\end{algorithm}

\begin{algorithm*}[t]
\caption{
TEE Initialization.
If the protocol succeeds,
TEE is loaded with keys, the dataset and the budget for answering DP queries.
}\begin{algorithmic}[1]
\Function{Init}{\,}
\State $(\datakey, \statekey, \sk_s, \vk_s, \vk_\scm) \gets \mathsf{GetKeys}(m)$~\Comment{Get signing \& encryption keys,
using TEE attestation based on measurement~$m$.}
\State Generate random $\nonce$
\State $\sigma_\scm, (\id, (\sigma_h, h)) \gets \SCM.\ReadState(\nonce)$~\Comment{Get state from SCM.}

\If{$\Verify(\vk_\scm, \sigma_\scm, \id, (\sigma_h, h), \nonce)$ or $\Verify\left(\vk_s, \sigma_h, (\id, h)\right)$ fail} 
    \State Abort \Comment{Latest state is either not from a valid SCM or has old $\nonce$ or the state was not from a valid TEE.}%
\EndIf

\State $\sigma, \encs \gets \ST.\Load(\statetext)$  \Comment{Load latest encrypted state from storage.}
\If{$\Verify(\vk_s, \sigma, \id, \encs)$ and $\hash(\encs) = h$}
\State Go to~\cref{line:finishinit} \Comment{State digests loaded from SCM and generated from the storage state are in sync.}
\Else
\If{$\Verify(\vk_s, \sigma, \id+1, \encs)$}
 \Comment{TEE may have crashed before SCM received its updated state.}
\State  $h \gets \hash(\encs)$ and $\sigma_h \gets \Sign\left(\sk_s, (\id, h)\right)$
\State Generate new random $\nonce$
	\State $\sigma_\scm, \msg \gets \SCM.\UpdateState(\nonce, \id, (\sigma_h, h))$
	\If{$\Verify(\vk_\scm, \sigma_\scm,  \msg, \nonce)$ fails or $\msg \neq \mathsf{Ack}$}
		\State Abort \Comment{Could not update the state at SCM (\eg because SCM has a newer state).}\label{line:abort3}
		\EndIf
\Else
  \State Abort \Comment{Loaded state cannot be authenticated with expected $\id$.}\label{line:diffids}

\EndIf

\EndIf 
\State$\query, \answer, \budget  \gets \Decrypt(\statekey, \encs)$\label{line:finishinit}
\If{$\id \neq 0$} \Comment{Check if answers to any queries have been computed.}
 	\State Send $\id, \query, \answer$ to the data analyst   \label{line:sendanswer1}
\EndIf 
\State $\enccd \gets \ST.\Load(\datatext)$
\State $\cd \gets \Decrypt(\datakey, \enccd)$ \Comment{Load and decrypt data}
\EndFunction
\end{algorithmic}
\label{alg:init}
\end{algorithm*}

\paragraph{TEE Initialization.}
The $\Init$ algorithm is run inside of a TEE every time a TEE is started.
During the protocol, TEE interacts with the untrusted server storage.
By the end of the initialization, TEE's local private memory contains the keys and the state components as described above.
The pseudocode of this protocol is shown in \cref{alg:init} and described below.

During initialization, TEE loads the signing and encryption keys either from sealed state or from a key vault,
ensuring that only TEE's running legitimate code can load the keys.
The TEE then loads its state from the server storage and requests a signed state digest from the SCM. 
For freshness, it requests the SCM to sign the $\nonce$ sent by the TEE.
If the signature and freshness checks succeed, 
the TEE tries to load the state. It compares the digest of the state it loads from the server storage with the digest
stored at the SCM. If the digests match and SCM's digest was uploaded earlier by a valid TEE (as verified using digital signatures with $\vk_s$),
the TEE will use the loaded state as its local state.
If the digests do not match,
TEE attempts to update the state digest at SCM. If the update is successful, it proceeds similar to above.
Otherwise it aborts as SCM must have been updated with a more recent state by another TEE.

The TEE checks if there is a query in the current state. 
If it is the first time TEE is started since the data owner setup the system,
there will be no pending query answers.
The TEE then loads and decrypts the dataset $\cd$.
If the dataset is too large to fit in TEE, this could be replaced with a procedure that
will allow for access to dataset $\cd$ in chunks 
(a common trick for constrained memory devices~\cite{usenix16} that will not affect a DP system to which the data will appear to be coming as one chunk).

\subsection{Replying to Queries}
A data analyst interacts with \dpsystem by sending queries to the server hosting it.
The server sends the query to an initialized TEE if there is one or initializes a new one. 
The TEE answers queries as long as there is enough budget left.
The pseudo-code for replying queries appears in~\cref{alg:reply}.

Let $\query$ be a query sent by the analyst.
TEE runs $\spec$ to determine the DP mechanism for query $\query$, $\mech$, and query budget,~$\epsilon$.

If the budget is enough to answer the query,~$\mech$ computes~$\answer$; otherwise, it returns~$\bot$.
Answering the query should be done without observable side channels
as otherwise information can be leaked if the TEE is interrupted and budget is not updated.
For example, if answer length differs depending on the dataset,
the answer needs to be padded to the length of the longest possible answer.

If the query can be answered without exceeding the budget, TEE sets new budget to~$\budget - \epsilon$.
Otherwise, budget stays the same.
Regardless of the query output,
TEE updates the state with an incremented state counter, the budget, last query and the answer.
It encrypts the state and stores it in the storage, before requesting SCM to update the digest.
If SCM can update the state digest (\ie SCM's last state identifier is one smaller than what TEE is trying to update it with)
and SCM's reply is fresh (\ie it is signed with TEE's freshly generated nonce),
the TEE sends the answer to the analyst and waits for further queries.
If the state at SCM cannot be updated, TEE aborts as its current state is out of sync with~SCM.

\begin{algorithm*}
\caption{
Pseudocode for answering queries inside of a TEE. 
Given a query, $\spec$ returns the DP mechanism to compute the answer and the budget required for this query.
After TEE is initialized 
in~\cref{alg:init}, it answers queries from a data analyst.
For each query~$q$, it associates an identifier and computes an answer.
If there is enough budget to answer the query, the budget is updated.
Otherwise, the answer is set to~$\bot$ and the budget stays the same.
Query, answer, its identifier and the budget are saved at the storage. 
If SCM accepts this state update, the answer is sent to the data analyst. 
Otherwise, TEE aborts as SCM may have a more recent state update.}
\begin{algorithmic}[1]

\Function{$\ReplyQueries(\query)$}{}
\State $\epsilon, \mech \gets \spec(\query,\budget)$ \Comment{Determine query mechanism and the privacy loss as per the owner's specification for this query.}
\If{$\budget - \epsilon \ge 0$}
\State $\answer \gets \mech(\cd)$ \Comment{Compute DP answer.}
\State $\budget \gets \budget - \epsilon$ 
\Else
\State $\answer \gets \bot$ \Comment{Not enough budget to answer this query.}
\EndIf

\State $\id \gets \id + 1$
\State $s \gets \query\|\answer \| \budget$
\State $\encs \gets \Encrypt(\statekey, s)$ and $h \gets \hash(\encs)$ \label{line:encrypt}
\State  $\sigma \gets \Sign(\sk_s, \id, \encs)$ and $\sigma_h \gets \Sign\left(\sk_s, (\id, h)\right)$
\State $\ST.\Store(\statetext, (\sigma, \encs))$ \Comment{Save the state together with the answer and the unique identifier.}
\State Generate random $\nonce$
\State $\sigma_\scm, \msg \gets \SCM.\UpdateState\left(\nonce, \id,(\sigma_h,h)\right)$ \label{line:store} \Comment{Attempt to update the state at counter $\id$.}
\If{$\Verify(\vk_\scm, \sigma_\scm,  \msg, \nonce)$ or $\msg \neq \mathsf{Ack}$ fail}
\State Abort \Comment{TEE could not update SCM's state. Abort as TEE is not in sync.}\label{line:abortreply}
\EndIf
\State Send $\id, \query, \answer$ to the data analyst\label{line:sendanswer2}
\EndFunction
\State Abort, if $\budget = 0$.
\end{algorithmic}
\label{alg:reply}
\end{algorithm*}

\subsection{Recovery}
If a TEE crashes, its local state in runtime memory is lost as well.
To recover, a new TEE needs to start from the state saved by the previous TEE in the untrusted persistent storage.
The protocol for starting a new TEE is the same as the one to initialize a TEE in~Algorithm~\ref{alg:init}.
Since the server is responsible for starting and running TEEs, it may crash TEEs, create several TEEs (e.g., trying to execute a fork attack)
or attempt to give obsolete states (e.g., trying to rollback the budget).
To this end, the protocol ensures that the state TEE loads from the server is indeed the latest.

The TEE checks if there is a query in the current state 
(\ie it is not the first time that a TEE of the DP system has been started).
If there is one, it sends previously recorded response $\answer$ to the data analyst.
This ensures that the analyst gets the answer if a TEE has crashed after it updated the budget, but before it sent the answer to the analyst.
This may result in the analyst getting multiple copies of the same answer if the TEE crashed after the answer was sent.
We made this choice to avoid budget depletion for the last query that has already consumed the budget.

The TEE then loads the dataset and waits for analyst's queries.
Note that \dpsystem can restart as long as TEE's requests to get the keys, SCM state and local storage are met by the server.
We discuss recovery further in this section.
 
\subsection{Discussion}
\label{sec:discussion}

\newcommand{\aux}{\mathsf{aux}}

\paragraph{Management of Privacy Budget.}
\dpsystem relies on privacy budget management provided by the data owner in functionality~$\spec$.
$\spec$ decides which mechanism $\mech$ to use to answer the query and how much privacy loss the query will incur.
Depending on the use case, $\spec$ could be instantiated with advanced privacy management schemes
that maximize the number of queries that can be answered by efficiently allocating privacy budget across queries and applications~\cite{cohere}.
Some of these DP schemes (\eg private multiplicative weights~\cite{dwork2014algorithmic}) keep a state between queries.
To support these mechanisms, we can abstract their state as an additional input $\aux$ to $\spec$ (\ie $\spec(\query, \budget, \aux)$) that
will need to be persisted in system's state~$s$ between queries.

We assumed that there is only one data analyst. The system can be extended with more.
For example, the data owner can setup the TEE code to answer
queries only from a pre-approved set of analysts by supplying their public keys.
Though this would not change our threat model that assumes collusion between analysts and the server,
it can assist in a more efficient budget allocation across the analysts~\cite{DBLP:journals/pacmmod/ZhangH23}.

\paragraph{Multiple Datasets.}
We can consider two settings. The first one, where a data owner has several datasets 
and wants to setup \dpsystem for each one. Our protocol can be adopted to support this 
by associating each dataset and its state with a label, where the label is unique for each dataset. 
Then, SCM needs to maintain states along with their labels such that one can query 
and update state digests by using their labels (\eg Nimble~\cite{angel2023nimble} provides labeling functionality).
The second setting is where a single system may need to access multiple datasets, 
each with its own budget (\eg in the Local DP setting). In this case, as for the first setting, 
\dpsystem will need to maintain a state per each dataset.
Additionally, it will also need to correctly load each state during the initialization 
and release query answers to a data analyst only after budgets for all datasets have been successfully updated at the~SCM.

\paragraph{Liveness vs.~Confidentiality.}
Consistency, Availability, and Partition tolerance (CAP) theorem states that one has to give up either safety or availability
in an unreliable distributed system~\cite{10.1145/343477.343502}.
Since our priority is to provide dataset privacy according to DP guarantees,
if \dpsystem crashes, it can recover as long as the server behaves according to the protocol.
As it is in the interest of the service provider to behave correctly, 
the scenario where a restarted TEE cannot access SCM or latest stored state is unlikely to happen. 
Nevertheless, the system can try to recover by relying on the intervention from the data owner. 
For example, if the dataset is lost, the owner could upload it to the server again; 
if the state is lost at the server, the data owner could retrieve it from SCM 
(in this case SCM would need to store encrypted state instead of its digest). 
However, if SCM is not available, the owner will have to decide on the value of the new budget, 
which could lead to a budget rollback since the owner does not know the value of the budget before the crash.

\section{\dpsystem Analysis}
\label{sec:analysis}

We now prove that \dpsystem provides liveness and DP confidentiality,
assuming TEE provides isolation guarantees and SCM provides state continuity guarantees as specified in~\cref{sec:scm}.

\subsection{Liveness}

Proofs of both properties rely on the integrity, authenticity and recency of the states loaded 
by a TEE summarized in the following lemma.

\begin{lemma}
\label{lemma:integrity}
Assuming the integrity and confidentiality properties of the TEE,
the integrity properties of SCM and standard cryptographic assumptions for signatures, authenticated encryption
and the hash function,
\dpsystem has the following state integrity properties:
\begin{enumerate}
\item if a TEE is successfully initialized, 
then it is loaded with a state that was created either by a valid TEE or by the data owner;
\item if a TEE is successfully initialized with identifier $\id$ 
then the last state stored by a valid TEE at SCM had identifier~$\id$
or~$\id-1$.
\end{enumerate}
\end{lemma}
\begin{proof}
The first property holds since the state stored at the storage module is signed using the key $\sk_s$ of the data owner.
By design only valid TEEs can access this key. The initialization in~\cref{alg:init} is successful only if state verification succeeds.

We now prove the second property. During initialization~(\cref{alg:init}) a TEE reads the state digest from SCM.
It checks that the received digest $h$ is from a valid SCM and fresh 
(checked using signature $\sigma_\scm$, verification key $\vk_\scm$ and $\nonce$). 
It also ensures that the digest was uploaded by a valid TEE using $\vk_s$.
It then compares the digest to the hash of the state loaded 
from the storage (based on the first property it is an authentic state but may be stale). 
If the digests are equal, initialization is successful and the second property holds.
If they are not equal, TEE tries to update the SCM. The SCM accepts an update 
only if its last recorded state digest has an identifier one less than the one sent by the TEE. 
It rejects an update if there are other entries associated with the same $\id$. 
Otherwise, it updates its state and sends an acknowledgement.
\end{proof}

\newcommand{\TEE}[1]{\emph{TEE}$_#1$\xspace}

\begin{lemma}
\label{lemma:live}
\dpsystem provides \emph{liveness} (Definition~\ref{def:crash-resilient}) in a non-adversarial environment,
that is, when a restarted TEE is given the latest state stored by a valid TEE and receives
a response from a valid SCM.
\end{lemma}

The proof of Lemma~\ref{lemma:live} appears in~\cref{app:proof}.

\subsection{DP Confidentiality}

\paragraph{Information released in \dpsystem.}
By design, only TEEs running data owner's code have access to the data encryption key~$k_\cd$.
Hence, information about the dataset is released by valid TEEs only.
TEEs store states to storage module, send updated digests to SCM, abort and return answers to a data analyst.
The first two actions do not reveal information about $\cd$ or computed answers $\answer$
since the state is stored encrypted and the digest is a hash of the encrypted state. 
If state size (e.g., the size of~$\answer$) is sensitive, one may consider padding it.
Aborts are associated with integrity of the state continuity and loaded states 
and do not depend on the content of the data or returned answers.
Hence, we only need to consider the information released by TEEs via replies
to queries.

We capture query replies via tuples $(\id, \query, \answer)$ sent by a TEE to a data analyst.
Let $\mathcal{T}$ be a sequence of these tuples.
We show that $\id$s in $\mathcal{T}$ follow a sequential order and are numbered $1$ to $n$ where
$n$ is the number of queries in~$Q$.
Due to TEE crashing and resending the answers to previously asked queries after recovery, 
$\mathcal{T}$ may have tuples with same~$\id$s. 
However, we show that if there is more than one tuple with the same $\id$, 
then all these tuples are equal (Lemma~\ref{lemma:cond1}). 
Hence, these entries do not reveal any new information and are copies of previously sent answers.
\ifCCS
Due to space constrains the proof of Lemma~\ref{lemma:cond1} appears in~\cref{app:proof}.
At a high level it shows that each tuple is associated with exactly one entry at SCM where each entry has a unique identifier.
\fi

\begin{lemma}
\label{lemma:cond1}
Let $\mathcal{T}$ be a sequence of tuples returned by a TEE.
Then (1) all identifiers in the sequence follow a sequential order, starting at 1
and (2) if more than one tuple has the same identifier then those tuples are equal.
\end{lemma}

\begin{lemma}
Assuming the integrity and confidentiality properties of a TEE,
the integrity properties of SCM and standard cryptographic assumptions for signatures, authenticated encryption
and the hash function, \dpsystem provides DP confidentiality (Definition~\ref{def:securedp}).
\end{lemma}

\begin{proof}
We first show how to construct transcript $T$ for \dpsystem and determine order $\pi$, as per~Definition~\ref{def:securedp}.
We then show the equivalence with the transcript of the ideal functionality.

Let $\mathcal{T}$ be all the tuples $(\id, \query, \answer)$ sent by a TEE to a data analyst (as defined in this section)
but after removing equal tuples.
We then augment each tuple with the budget~$\budget$
stored in a TEE (\eg runtime memory) at the time it sent this tuple.
Transcript~$T$ is then constructed as follows: $T[\mathcal{T}.\id] = (\mathcal{T}.\query, \mathcal{T}.\answer, \budget)$.
The index positions associated with queries in $T$ form an order $\pi$ on $Q$ that we use for~Definition~\ref{def:securedp}.

Given properties of the sequence of replies $\mathcal{T}$ (Lemma~\ref{lemma:cond1}) and by construction above,
we know that both transcripts are of the same length and queries are equal: $T[i].\query = T_\mathsf{ideal}[i].\query$.
Hence, we need to show that the corresponding budgets are equal and answers are either both equal to $\bot$
or are from the same distribution (condition of Definition~\ref{def:securedp}).
Note that if budgets are equal then the condition on the answers will hold since both the ideal functionality and the code
in the TEE use the same mechanism~$\mech$ to compute an answer (albeit with different randomness).
Moreover, both use deterministically computed $\epsilon$ to determine if a query exceeds the budget or not. Hence, both functionalities will either return~$\bot$ or answers from the same distribution.
Hence, in the rest of the proof we show that budgets across both transcripts are the same.

We prove that budgets are equal by showing that the answer and the (updated) budget for entry with identifier~$\id$ use the budget of entry $\id-1$ or owner's initial budget, if $\id=1$. 
Since budget update is performed by the same
code as the one from the ideal functionality, this would lead to an equivalent budget update and answer computation
for this entry.

In the base case, when no queries have been answered, there is only one state that is available in the system
and can be successfully loaded into a TEE; it is the state generated by the owner: $\bot\|\bot\|\bm{B}$ with $\id = 0$. 
In this case SCM is not updated as it stores the same state loaded by the owner. 
According to~Lemma~\ref{lemma:integrity},
 if TEE is initialized, it is loaded with owner's budget and $\id=0$. 
When query associated with $T[1]$ arrives, 
$\ReplyQueries$ uses $\bm{\budget}$ to determine the answer 
and the next budget, and  $\id$ to increment the identifier. 
The answer is sent (and hence appears in the transcript) only if the state was successfully updated~at~SCM.

For the general case, we prove that the only state that TEE can use to answer query 
at entry $T[i]$ must be from the state associated with an answered query $T[i-1]$.
First we show that the state that a TEE uses to determine an answer for any query 
must have an entry associated with it in $T$ (\ie there must have been an answer produced by the TEE). 
This is true, since a state is stored by a TEE and updated at SCM, before the answer is sent. 
Even if a TEE crashes after it has updated both storage and SCM but before it sends an answer, 
the new TEE will be loaded from the saved state and will send the saved answer before proceeding 
to reply other queries (\cref{line:sendanswer1} in~\cref{alg:init}). 
Since TEE associated identifier $i$ with this particular query, 
its previous state must have been associated with identifier $i-1$. 
The states have the same identifiers as the answers sent at that state, hence, 
the answer must be associated with identifier $i-1$. Since identifiers are unique in the transcript, 
the previous answer must have been for the entry $T[i-1]$.
\end{proof}

\section{\dpsystem Evaluation}
\label{sec:eval}

We evaluate the performance of \dpsystem
w.r.t.~a na\"ive DP protocol that provides no security beyond DP guarantees.

\subsection{Experimental Setup}

The experiments in this section were run
on Intel Xeon E-2288G CPUs that support SGX, running a 64-bit Ubuntu Linux 20.04.4 
with kernel version 5.15.0-1034 (same as in~\cref{sec:sgx-attack}).

\paragraph{\dpsystem.}
We implement \dpsystem as described in~\cref{sec:proto} in C++.
TEE and SCM are the two main components of \dpsystem.
TEE is instantiated using CPU with
Intel SGX capabilities, processes DP queries
and updates system states as in~\cref{alg:reply} inside an enclave.

In the experiments we choose to instantiate SCM with 
a distributed system of nodes with Intel SGX capabilities in a LAN environment of a data center,
each running on a separate machine to protect against a single-point failure.
Unless otherwise specified SCM has 3 nodes.
The SCM is adapted from the distributed system part of Narrator~\cite{niu2022narrator},
and accommodates the functionalities required by our protocol,
such as the signing functionality used in $\InitState$ of~\cref{alg:scm}.
Narrator adopts a consistent broadcast protocol~\cite{cachin2011introduction,reiter1994secure},
and guarantees that if a node 
fails, it can be restarted and synchronized with the latest digest
(\ie the digest with the highest $\id$) from the remaining nodes.
We use AES-GCM for encryption, SHA256 for hash, and RSA-SHA256 for signatures,
which are provided by Intel SGX SDK.

\newcommand{\Avg}{\texttt{Avg}}
\newcommand{\Var}{\texttt{Var}}
\newcommand{\Cor}{\texttt{Cor}}

\paragraph{\naivedp.}
As a baseline we use a system that runs DP code on the machine with the same specification as \dpsystem but
does not use TEEs, state continuity, encryption and signatures.

\paragraph{Datasets and queries.}
As datasets, we use the PUMS dataset~\cite{pums} (1.2 million records, 44\,MB) and the test dataset of DIGIX~\cite{digix} (2~million records, 307\,MB).
Since DIGIX is 307\,MB and the maximum SGX enclave memory is around 190\,MB,
we divide and seal DIGIX into two data chunks so that each can fit into the enclave.
During queries, the code abstracts the data as a single piece and 
loads records as requested by a DP mechanism (e.g., for average query, Gaussian mechanism may access one record at a time).
Both PUMS and DIGIX have an age column, which is used in most of our queries.

We test six queries, which range in complexity and output sizes:\begin{enumerate}
  \item Average age.
  \item Variance of age.
  \item Correlation: the Pearson correlation coefficient 
  between age and income columns (for the PUMS dataset).
  \item GroupBy:
  the averaged incomes for different age groups of the PUMS dataset.
  The number of groups ranges in \{20, 50, 100\} and corresponds to the query output size.
  For example, if the group size is 100 (i.e., group by each age year), 
  then the output consists of 100 floating-points.
  \item Shuffle DP for reporting age group:
  the computation proceeds by shuffling the age column,
  bucketing the age to one of 10 groups and randomizing this value
  using Alg.~1 from~\cite{balle2019privacy}. The output is the randomized age group of each
  record in the shuffled order.
  Hence, query output size is equal to the input size (\eg 1.2M for the PUMS dataset).
  \item Private Multiplicative Weights (PMW):
  the query takes input~$A$ and computes \%records with age of $A$ or higher.
  We use Multiplicative Weights Update Rule (Alg.~5 from~\cite{dwork2014algorithmic}).
  The query maintains a list of 100 weights (randomly initialized) used to estimate a result (one for each age year).
  The query returns the estimated result based on this list
  if the relative error is within some target accuracy and no privacy budget is consumed.
  Otherwise, the real result protected with Laplace noise is returned,
  budget is updated (as determined by the owner at the setup) 
  and weights are updated. 
  The list of weights is persisted in the~state.
\end{enumerate}

For queries protected with the Gaussian and Laplace DP mechanisms,
each query is assigned with budget $\epsilon=1$,
and its sensitivity is estimated w.r.t.~the range of input values 
(\eg 100 for the age column) and query type.
All queries except  GroupBy and Shuffle DP return one value as an output.
Besides storing the updated budget and query output, PMW also stores the weight list in its state.

\begin{table*}[t]
\centering
\renewcommand{\arraystretch}{1.15}
\caption{
The average time (in ms) and throughput (\ie the number of queries answered per second)
to answer Average, Variance, Correlation, GroupBy, Shuffle DP, and PMW queries
for \naivedp and \dpsystem on the PUMS dataset. The relative overhead and throughput of \dpsystem
wrt.~\naivedp is reported in brackets.
The timing results are amortized over 1000 queries including the initial cost of loading of a dataset.
}

\begin{tabular}{c|c|c|c|c|c|c|c}
                                                           &       \multirow{2}{*}{System}                                    & \multicolumn{6}{c}{Query Type}                                                                                                                                                                                                                                                                                                                                                                                                                     \\ \cline{3-8} 
                                                                                    &                                                                    & {Average}                                             & {Variance}                                              & {Correlation}                                           & {GroupBy}                                               & {Shuffle DP}                                           & PMW                                                  \\ \hline\hline
\multirow{3}{*}{Time (ms)}                & NaiveDP                                                            & {3.12}                                                & {10.12}                                                 & {27.56}                                                 & {5.7}                                                   & {44.6}                                                 & 2.42                                                 \\ \cline{2-8} 
                                                                                    &\multirow{2}{*}{ElephantDP}  & 6.18 & 12.04 & 29.09 &18.24&69.3&6.94 \\ 
& &  (2$\times$)  & (1.2$\times$)        & (1.1$\times$)         & (3.2$\times$)                & (1.5$\times$)                 & (2.8$\times$)        \\                         
                                                                                    \hline\hline
                                                                                    
\multirow{3}{*}{\parbox{2cm}{\centering Throughput (\#queries/s)}} & NaiveDP                                                            & {321}                                                 & {99}                                                    & {37}                                                    & {176}                                                   & {23}                                                   & 414                                                  \\ \cline{2-8} 
                                                                                    & \multirow{2}{*}{ElephantDP}  &160 &83  & 35  &55    & 15  &144\\
                                                                                    & &  (0.5$\times$)  & (0.8$\times$)        & (0.9$\times$)         & (0.3$\times$)                & (0.7$\times$)                 & (0.3$\times$)        \\                   
\end{tabular}
\label{tab:query-eval}
\end{table*}

\subsection{Evaluation}

We measure the runtime overhead of answering a single query 
in a system that loads the PUMS dataset and answers 1000 queries, 
hence, the cost of loading the initial dataset is amortized.
  All measurements are averaged over 100 runs.

We report the exact times and throughputs in~\cref{tab:query-eval} for all queries,
including overhead and performance evaluations.

We observe that queries with shorter query answer time by \naivedp 
(\eg Average, GroupBy, and PMW) usually have bigger overheads 2--3.2$\times$ by \dpsystem,
while queries with longer query answer time 
(\eg Variance, Correlation and Shuffle DP) have smaller overheads 1.1--1.5$\times$.
That is, the overhead of \dpsystem compared to \naivedp is lower 
when computing complex queries.
\dpsystem can process between 15 to 160 queries per second, depending on the complexity of the query.

To understand what adds to an overhead we measure:
\begin{description}
  \item $t_Q$: the time it takes to compute the DP answer to a query 
  (both systems incur this cost);
  \item $t_L$: the time it takes to load and decrypt the dataset and system state
  (\naivedp incurs only the loading cost, \dpsystem incurs both the loading and decryption cost);
  \item $t_S$: the time it takes to retrieve and update the state at SCM (incurred by $\dpsystem$).
\end{description}
We report these measurement on all queries in~\cref{fig:overhead}.

\begin{figure}[h]
\includegraphics[width=0.3\textwidth]{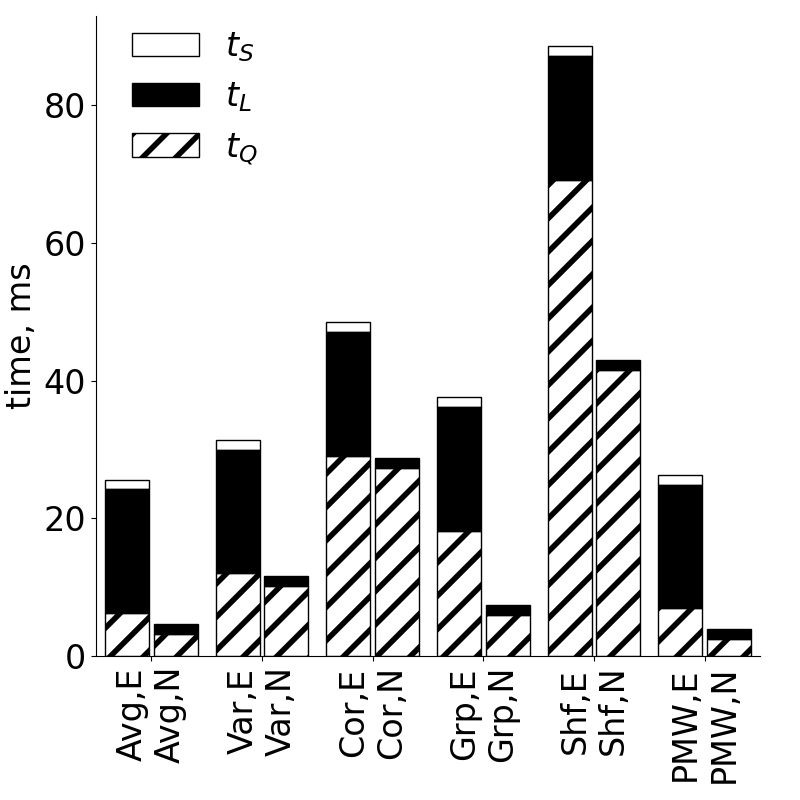}
\centering
\caption{
  The overhead for a single query Average (Avg), Variance (Var), Correlation (Cor),
  GroupBy (Grp), Shuffle (Shf), and Private Multiplicative Weights (PMW) 
  on the PUMS dataset using \dpsystem (E) and \naivedp~(N).
}
\label{fig:overhead}
\end{figure}

We observe that
most of the overhead comes from initializing the enclave and loading the data.
However, the dataset loading cost~$t_L$ is a one-off cost and can be amortized across queries, 
while~$t_Q$ and~$t_S$ are incurred per query.

Finally, we observe that \dpsystem incurs a small memory overhead from storing a single query, its answer, 
the current budget (e.g., an integer or a floating point value), and the associated signature at the server. 
This state size is proportional to the query, answer representation and internal state size (e.g., for PMW).
For statistical queries above these are well under 1KB, except for Shuffle DP, whose query answer is proportional to the data size, akin to local DP.

\paragraph{Summary.}
Overall, the experiments show that the overheads of our non-optimized implementation of \dpsystem are acceptable,
particularly in the case of complex queries. 
We note that the overhead for large datasets  comes mostly from data loading 
--- a cost that would be incurred by any system that uses TEEs for processing data. 
\dpsystem, however, provides much stronger guarantees by
protecting the budget from, for example, rollback and fork attacks which can compromise DP guarantees.

\subsection{Ablation Study}
\label{sec:ablation}
In this section we evaluate separately how different system parameters influence the performance of \dpsystem,
including the number of nodes in SCM, input size, and output size.

We first test increasing \emph{the number of nodes in SCM},
from 3 to 5 and 7, which would allow SCM to recover if more nodes fail.
We observe that when a system has more nodes,
the state update cost, $t_S$, increases from an average of 1.4\,ms to 2.3\,ms and 3.1\,ms respectively,
due to the cost of broadcast.
We note that SCM cost does not depend on a dataset or a query type as \dpsystem's state consists of a hash and a signature
which are both constant.

To explore \emph{the influence of input size},
we compare performance of
Average, Variance, and Correlation queries on the PUMS dataset and a bigger dataset DIGIX.
The results are presented in~\cref{tab:digix-eval}.
We observe that the time to compute on larger input size scales linearly
and the overhead stays the same.

\begin{table}[h]
\centering
\addtolength{\tabcolsep}{-0.3em}
\renewcommand{\arraystretch}{1.15}
\caption{
The average time (in ms) to answer queries Average, Variance, and Correlation
for \naivedp and \dpsystem on datasets of different sizes: PUMS (44MB) and DIGIX (307MB).
The timing results are amortized over 1000 queries including the cost of loading of a dataset.
The relative overhead of \dpsystem
wrt.~\naivedp is reported in brackets.
}
\begin{tabular}{c|c|rr}
\multirow{2}{*}{Dataset} & \multirow{2}{*}{Query} & \multicolumn{2}{c}{Time (ms)} \\ \cline{3-4}
                         &                        & \multicolumn{1}{r|}{NaiveDP}            & ElephantDP                    \\ \hline\hline
\multirow{3}{*}{PUMS}    & Average                & \multicolumn{1}{c|}{3.1}               & 6.2 (2$\times$)               \\
                         & Variance               & \multicolumn{1}{c|}{10.1}              & 12.0 (1.2$\times$)           \\
                         & Correlation            & \multicolumn{1}{c|}{27.6}              & 29.1 (1.1$\times$)           \\ \hline\hline
\multirow{3}{*}{DIGIX}   & Average                & \multicolumn{1}{c|}{5.2}               & 10.4 (2$\times$)              \\
                         & Variance               & \multicolumn{1}{c|}{16.9}              & 20.1 (1.2$\times$)           \\
                         & Correlation            & \multicolumn{1}{c|}{46.1}              & 48.2 (1.1$\times$)                         
\end{tabular}
\label{tab:digix-eval}
\end{table}

We test GroupBy and Shuffle DP queries with \emph{different output sizes}.
\ifCamera
The results are presented in~\cref{tab:output-eval} (with detailed breakdown of overheads
provided in~\cref{fig:output-overhead} of~\cref{app:abla-res}).
\else
The results are presented in~\cref{tab:output-eval} (with detailed breakdown of overheads
provided in Appendix in~\cref{fig:output-overhead}).
\fi
We observe that
as the query answer time increases with increasing output size for both systems,
the overhead of loading the dataset is amortized.

\begin{table}[h]
\centering
\caption{
The average time (in ms) to answer queries GroupBy and Shuffle DP
for \naivedp and \dpsystem on the dataset PUMS.
The output size  ranges in \{20, 50, 100\} for the GroupBy query and in \{120K, 600K, 1.2M\}
 for the Shuffle DP query.
The results are amortized over 1000 queries including the cost of loading the dataset.
}
\begin{subtable}[c]{.235\textwidth}
\centering
\addtolength{\tabcolsep}{-0.3em}
\renewcommand{\arraystretch}{1.15}
\caption{GroupBy Query \label{tab:output-a}}
\begin{tabular}{p{0.95cm}|p{1.1cm}|R{1.4cm}}
\multirow{2}{*}{\parbox{0.95cm}{\centering Output Size}} &  \multicolumn{2}{c}{Time (ms)}\\\cline{2-3}
& {\centering{NaiveDP}} & ElephantDP\\
\hline
\hline
\centering 20          & \centering 4.4    &  $17.6$ (3.9$\times$) \\
\centering 50          & \centering 4.9    &  $17.8$ (3.6$\times$) \\
\centering 100         & \centering 5.7     &  $18.2$ (3.2$\times$) \\
\end{tabular}
\end{subtable}
 \hspace{-0.05cm}
\begin{subtable}[c]{.235\textwidth}
\addtolength{\tabcolsep}{-0.3em}
\renewcommand{\arraystretch}{1.15}
\caption{Shuffle DP Query \label{tab:output-b}}
\begin{tabular}{p{0.95cm}|p{1.1cm}|R{1.4cm}}
\multirow{2}{*}{\parbox{0.95cm}{\centering Output Size}} &  \multicolumn{2}{c}{Time (ms)}\\\cline{2-3}
& {\centering{NaiveDP}} & ElephantDP\\
\hline
\hline
\centering 120K          & \centering 3.9      & $7.1$ (1.8$\times$) \\
\centering 600K          & \centering 21.6     & $34.7$ (1.6$\times$) \\
\centering 1.2M          & \centering 44.6      & $69.3$ (1.5$\times$) \\ 
\end{tabular}
\end{subtable}
\label{tab:output-eval}
\end{table}

\section{Related Work}
\paragraph{Secure Databases.}
Secure database systems aim to provide end-to-end confidentiality such that 
only the analyst can learn query outputs while data remains encrypted at the service provider~\cite{10.1145/3318464.3386141,10.1145/2043556.2043566,StealthDB}.
In such systems the analyst is trusted and receives true query answers, as opposed to our setting where answers
are protected from untrusted analysts with DP guarantees.
 Moreover, such systems do not consider an active attacker who can reset persistent storage.
That is, a database can be easily reset to a previous state if the attacker makes a copy of the persisted files.
Note that database logging, on its own, is not sufficient to protect against active attackers. For example,
an attacker can truncate the log (even if it is encrypted and integrity protected)
and give it to the database upon recovery, thereby rolling back its state. Hence, any database  set in a system with an active attacker has
to rely on a state continuity module. Indeed, EnclaveDB~\cite{8418608} relies on ROTE~\cite{matetic2017rote} to protect against rollback attacks.
To protect the state against fork attacks (same as those considered in this paper), 
EnclaveDB relies on analysts (referred to as clients in EnclaveDB). 
Since analysts are not trusted in \dpsystem, we propose a protocol between a TEE and SCM to ensure that only one TEE is updating the~budget.

\paragraph{Differential Privacy.}
Differential Privacy aims to bound information leakage released from answers to queries on a dataset.
In this paper, we considered a setting where a dataset owner outsources their data to an untrusted cloud and want to ensure
that analysts querying the data obtain answers with DP guarantees. We showed that it is crucial to protect DP budget
and presented \dpsystem that can achieve DP guarantees when DP code and budget are stored in an untrusted setting (\eg cloud).
Existing work has demonstrated that DP can also be vulnerable to other attacks in practice.
For example, sampling from real-valued distributions, as required by Laplace and Gaussian mechanisms, 
is not generally possible when one is limited to floating-point arithmetic, 
leading to privacy leakage~\cite{mironovFP,jin2022we,holohan2021secure,DBLP:conf/innovations/BalcerV18}.
Precision-based attacks~\cite{haney2022precision} exploit a pattern for units 
in the last places of the floating-point outputs from a DP system.
Finite-precision arithmetic
can also lead to underestimation of function sensitivity required for many DP mechanisms~\cite{10.1145/3548606.3560708}.
These attacks allow one to distinguish DP outputs of neighboring datasets, and thus invalidate DP guarantees.
Timing side-channels have also been exploited in DP systems~\cite{jin2022we,andrysco2015subnormal,haeberlen2011differential}.
To our knowledge, attacks against privacy budget (\eg via rollback) have not been considered.

Protecting DP with TEEs has been proposed in several prior works.
\textsc{prochlo}~\cite{bittau2017prochlo} collects encrypted user data 
and uses this data inside of TEEs with local DP methods.
Allen~\textit{et al.}~\cite{allen2019algorithmic} propose DP mechanisms protected against memory side-channels
when executed inside of a TEE.
DuetSGX~\cite{duetsgx} proposes to use a typechecker to verify 
that each query satisfies differential privacy before running it with Intel SGX. 
DuetSGX explicitly maintains a budget, though in runtime memory. 
Hence, in case of a crash, a TEE cannot be recovered with the latest budget 
and one may not know how many queries were answered.
Among the systems we outline in~\cref{tab:budget-impl} 
only PySyft~\cite{Ziller2021} explicitly stores the budget in a database. 
However, since it runs with no TEE guarantees and the file is not protected with either state continuity, 
integrity or confidentiality guarantees, it is susceptible to budget attacks.
In parallel to our work, Google proposed to compute advertising reports with DP guarantees
within TEEs as part of their Privacy Sandbox system~\cite{secureagg}, 
with a process for budget recovery post-disasters expected to be added in Q2 2024.

\paragraph{State Continuity.}
There are many systems that provide state continuity either by relying on monotonic counters 
or a distributed network of nodes~\cite{strackx2016ariadne,niu2022narrator,angel2023nimble,matetic2017rote}.
As we discuss in~\cref{sec:proto}, using them off the shelf for a DP system needs to be done carefully to ensure DP guarantees are provided. 
For example, one needs to determine what to save in a state and when to send the query answer to the data analyst.
We prove that our design guarantees liveness and DP confidentiality by relying on an SCM
that provides the functionality in~\cref{alg:scm}. Hence, any state continuity that fits this abstraction can be used within \dpsystem.

The closest to our work is Memoir~\cite{parno2011memoir}, a generic framework to guarantee state continuity using TPMs 
that allows recovery of a state by storing the history of state updates. 
Compared to \dpsystem, Memoir was not designed for DP and does not use modern TEEs (\eg Intel SGX).
Different from Memoir, \dpsystem updates persistent storage before recording state update at SCM. To this end, 
we require that the computed answer, not the query, to be recorded in the state. 
Thus we avoid Memoir's requirement to make the DP code deterministic.
We ensure that the answer is sent to the data analyst every time a TEE is restarted. 
Hence, when recovering from a crash, it is sufficient to recover the latest state without recomputing the query 
 or the history of updates 
--- leading to an arguably simpler protocol. On a technical note, 
the above choices result in a different protocol that we prove to be secure.
Our implementation also uses multiple TEEs to implement SCM as opposed to a single TPM that maybe subject to failure.

\section{Conclusion}
We propose \dpsystem, a DP system that provides
a secure instantiation of global DP in a remote untrusted setting.
\dpsystem design includes a crucial component of a DP system 
--- secure tracking of the privacy budget allocated for a dataset.
To this end, \dpsystem ensures that this budget is updated correctly 
after each query and the query is answered only if there is sufficient budget left.
Compared to previous work that do not provide state continuity for the budget,
\dpsystem ensures that the budget cannot be rolled back, replayed or forked.
\dpsystem uses a state continuity module and TEEs to guarantee liveness and DP confidentiality.
We use Intel SGX to instantiate our system and show that the overhead of the system 
from maintaining the state is small, while the cost of loading data can be amortized across multiple queries.

\section*{Acknowledgments}
This work has been supported by
  the joint CATCH MURI-AUSMURI; %
  the ARC Discovery Project number DP210102670; %
and
  the Deutsche Forschungsgemeinschaft (DFG, German Research Foundation) under Germany's Excellence Strategy - EXC 2092 CASA - 390781972. %
The first author is supported by the University of Melbourne 
  research scholarship (MRS) scheme.

\bibliographystyle{ACM-Reference-Format}
\bibliography{references}

\newpage
\appendix

\section{Proof of Lemmas}
\label{app:proof}

\newcommand{\thmtxt}{
\dpsystem provides \emph{liveness} (Definition~\ref{def:crash-resilient}) in a non-adversarial environment,
that is, when a restarted TEE is given the latest state stored by a valid TEE and receives
a response from a valid SCM.
}

\newtheorem*{T2}{Lemma~\ref{lemma:live}}
\begin{T2}
    \thmtxt
\end{T2}

\begin{proof}
Consider a scenario where a TEE running \dpsystem has crashed.
Let $\encs$ be the last encrypted state stored in the storage module.
Let us refer to the TEE that attempts to initialize as \TEE1 and \TEE2 be the TEE that computed $\encs$.
We need to show that a \TEE1, initialised using $\Init()$ protocol in \cref{alg:init},
loads $\encs$.
Our proof builds on the result of Lemma~\ref{lemma:integrity}:
the state loaded in a successfully initialized TEE must be most recent and come from a valid TEE or the data owner.

To complete initialization the protocol needs to reach \cref{line:finishinit}.
There are three Aborts that can prevent a successful initialization.
The first abort is due to a failed signature and freshness verification. 
These checks should succeed if the environment is not adversarial.
That is, if the correct SCM replied to the TEE in real time (\ie
the~$\nonce$ generated by \TEE1 was used and SCM's reply was signed using $\sk_\scm$), 
and the state in the reply was sent to SCM 
by a valid TEE (\ie verified using $\sigma'$ and~$\vk_s$).

The second abort happens when SCM's state does not match the one loaded in the TEE.
SCM's state can diverge from storage's state either because \TEE2 crashed after updating storage state 
but before SCM has received the state update or when there is another TEE that may have progressed SCM's state.
In the former case the update of the SCM should succeed if SCM has not been updated 
since \TEE2 has crashed and hence the latest identifier it stores must be one less than~$\id$ loaded by~\TEE1. 
The latter case would imply that there is a newer state available in the storage module, 
since storage module receives the state update before the SCM. 
However, this would contradict the assumption that $\encs$ is the latest state.

The third abort happens when loaded state cannot be verified with the expected identifier~$\id$, 
that is stored at SCM, or~$\id+1$, if a TEE was one state ahead of SCM. 
However, this can happen if storage module provides corrupted or stale states as compared to what is stored at SCM.
\end{proof}

\newcommand{\thmtext}{
Let $\mathcal{T}$ be a sequence of tuples returned by a TEE.
Then (1) all identifiers in the sequence follow a sequential order, starting at 1
and (2) if more than one tuple has the same identifier then those tuples are equal.
}

\newtheorem*{T1}{Lemma~\ref{lemma:cond1}}

\begin{T1}
    \thmtext
\end{T1}

\begin{proof}
We prove the lemma by showing that if the answer was sent by a TEE then there must
be exactly one entry in SCM associated with the same identifier that was generated by a valid TEE.
Note that the tuple is sent by a TEE (in $\Init$ and $\ReplyQueries$) only if the state with the same $\id$ was
successfully updated by the SCM. Hence, by construction there is a state associated with each answer with the same~$\id$.
Since SCM updates the state only if $\id$ sent by a TEE is one greater than
the one stored by SCM, all valid SCM states are associated with sequential $\id$s.
SCM's state is initialized by the data owner with $\id =0$. Hence, the only update that can be successful after initialization
must be for $\id =1$ and as a result the first answer is associated with identifier~1.

Copies of the same answer are generated by TEE if the TEE crashes after sending the reply in~\cref{line:sendanswer1} 
or~\cref{{line:sendanswer2}} in~\cref{alg:init} and~\cref{alg:reply}, respectively, 
but before it processes the next query and updates SCM's state accordingly.
Once an answer is computed, it is associated with an $\id$, and stored encrypted in the storage module --- becoming part of the state.
If TEE crashes before it updates the state with the next query, it will resend the identifier, 
query and answer as stored in the state for the previous query. 
This resent information will be the same since it is retrieved from the correctly recovered state as per~\cref{lemma:integrity}.
\end{proof}

\section{Additional evaluation results}
\label{app:abla-res}

\cref{tab:node-eval} lists time overhead (in ms) of sate update ($t_S$) for
different number of nodes in SCM (\cref{sec:ablation}).

\begin{table}[h]
\centering
\addtolength{\tabcolsep}{-0.3em}
\renewcommand{\arraystretch}{1.15}
\caption{
The time overhead (in ms) of sate update ($t_S$) for different number of nodes in SCM.
All timing measurements are averaged over 100 runs.
}
\begin{tabular}{c|c|c|c}
\multicolumn{1}{c|}{Node Number} & \multicolumn{1}{c|}{3}   & \multicolumn{1}{c|}{5}     & \multicolumn{1}{c}{7}    \\ \hline\hline
Time (ms)                        & 1.4                      & 2.3                        & 3.1                        \\
\end{tabular}
\label{tab:node-eval}
\end{table}

Figure~\ref{fig:output-overhead} displays detailed overhead when varying the output size 
in Table~\ref{tab:output-eval} in~\cref{sec:ablation}.

\begin{figure}[H]
\centering
\begin{subfigure}{0.23\textwidth}
\centering
    \includegraphics[width=0.8\linewidth]{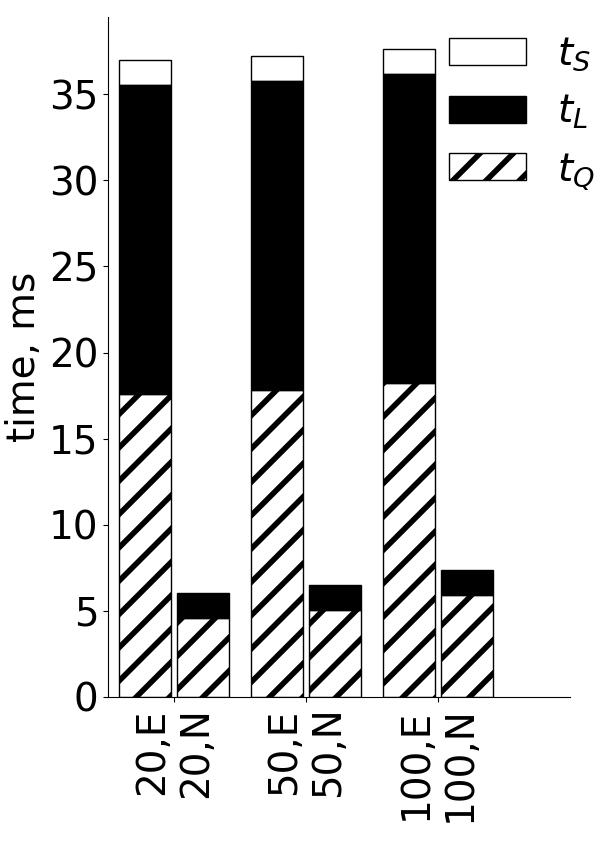}
    \caption{GroupBy Average Query}
    \label{fig:oh-gr}
\end{subfigure}
\begin{subfigure}{0.23\textwidth}
\centering
    \includegraphics[width=0.8\linewidth]{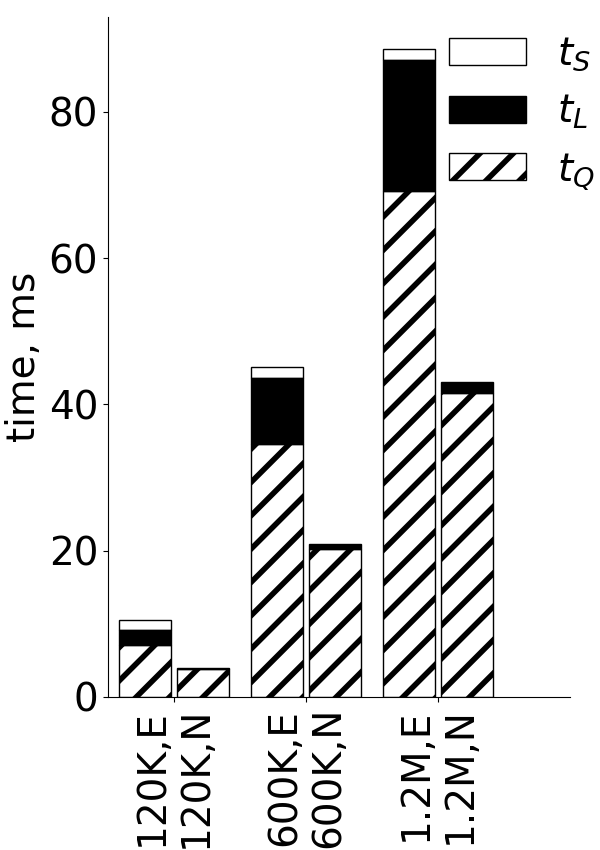}
    \caption{Shuffle DP Query}
    \label{fig:oh-shf}
\end{subfigure}
\caption{
  The time overhead (in ms) of answering queries GroupBy and Shuffle DP 
  by \naivedp (N) and \dpsystem (E) on the PUMS dataset. 
  The time comprises of: dataset loading ($t_L$),
  answering a query ($t_Q$) and updating the state at SCM by \dpsystem ($t_S$).
  (a) GroupBy query on the income grouped by the age (same as in Table~\ref{tab:output-a}).
  The number of groups ranges in \{20, 50, 100\} and corresponds to query output size.
  (b) Shuffle DP query (same as in Table~\ref{tab:output-b}) on the age column.
  The output size ranges in \{120K, 600K, 1.2M\}.
  All measurements are averaged over 100 runs.
}
\label{fig:output-overhead}
\end{figure}

\end{document}